\documentclass[runningheads]{llncs}

\usepackage[firstpage]{draftwatermark}
\SetWatermarkText{\hspace*{4.5in}\raisebox{6.3in}{\includegraphics[scale=0.1]{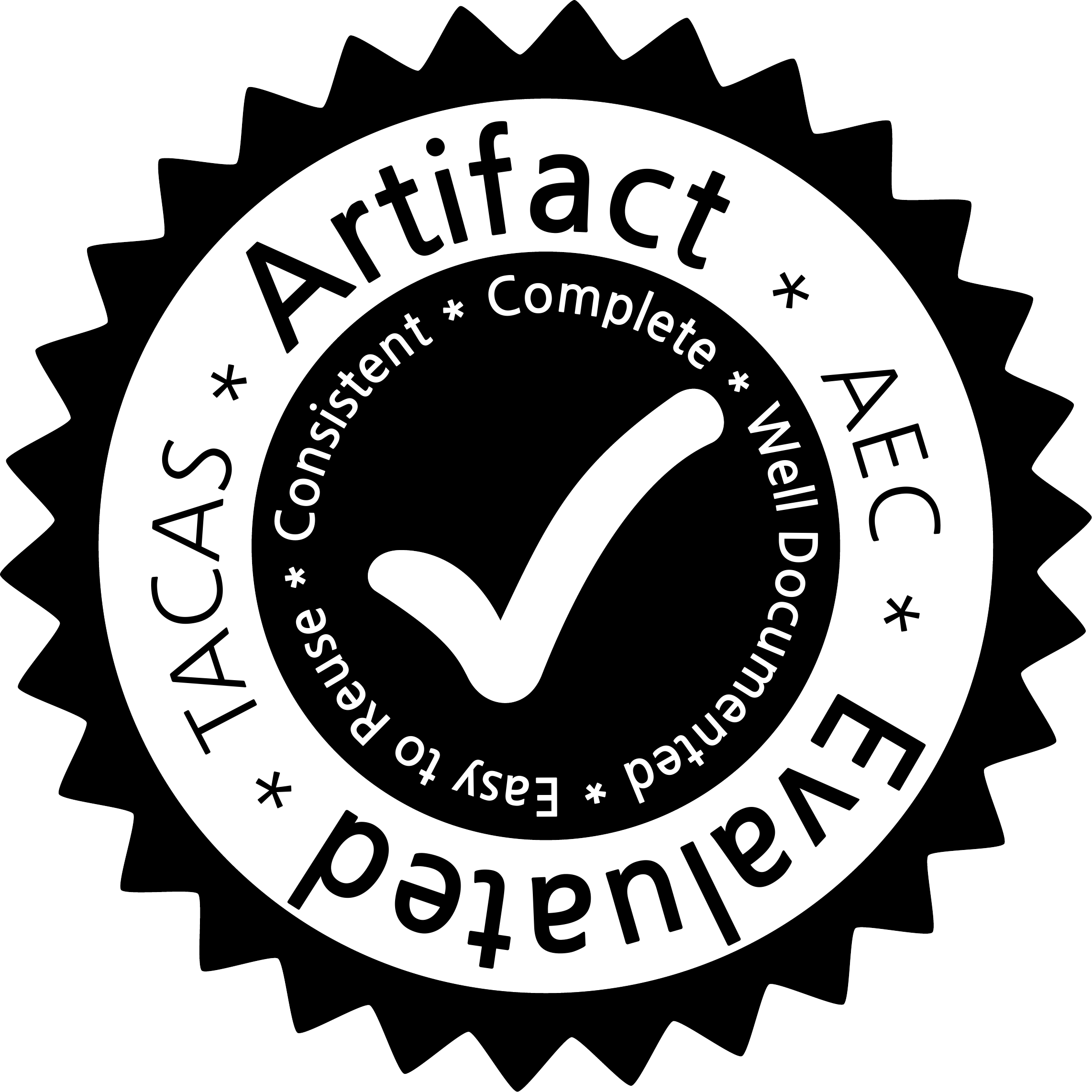}}}
\SetWatermarkAngle{0}

\usepackage{graphicx}
\usepackage{mathtools}
\usepackage{ltl}
\usepackage[vlined,linesnumbered]{algorithm2e}
\usepackage{subcaption}

\usepackage{underscore}
\usepackage{cite}

\usepackage{hyperref}
\usepackage[utf8]{inputenc}
\usepackage{pgfplots}
\usepgfplotslibrary{external}
\tikzset{external/export=false}
\usepgfplotslibrary{groupplots}
\pgfplotsset{compat=newest}

\usepackage{wrapfig}

\newcommand{\powerset}{\mathcal{P}}

\newcommand{\card}[1]{{|#1|}}
\newcommand{\ldot}{\mathpunct{.}}
\newcommand{\bool}{\mathbb{B}}

\newcommand{\dom}{\mathrm{dom}}

\newcommand{\AP}{\mathit{AP}}

\newcommand{\imp}{\mathbin{\rightarrow}~}

\newcommand{\defeq}{\mathrel{\coloneqq}}

\newcommand{\set}[1]{\ensuremath{\{#1\}}}

\newif\ifcomment

\newcommand{\ltl}{\text{LTL}}
\newcommand{\hyperltl}{\text{HyperLTL}~}
\newcommand{\secltl}{\text{SecLTL}}

\newcommand{\langL}{\mathcal{L}}
\newcommand{\ap}{\text{AP}}
\renewcommand{\models}{\vDash}
\newcommand{\nmodels}{\nvDash}
\newcommand{\modelsltlfin}{\models_{\ltl_\mathit{fin}}}

\newcommand{\tracebox}[1]{\framebox[7ex][l]{#1}}

\newcommand{\pathvars}{\mathcal{V}}
\newcommand{\pathassignfin}{\Pi_\mathit{fin}}

\newcommand{\lang}[1]{\langL(#1)}
\newcommand{\langx}[2]{\langL_{#2}(#1)}
\newcommand{\langxx}[3]{\langL^{#2}_{#3}(#1)}
\newcommand{\proj}[3]{#1|^{#2}_{#3}}
\newcommand{\enc}[1]{enc(#1)}
\newcommand{\encx}[2]{enc_{#1}(#2)}
\newcommand{\encn}[2]{enc^{#2}(#1)}
\newcommand{\encnx}[3]{enc^{#3}_{#1}(#2)}

\newcommand{\varposo}[0]{v^{+}}
\newcommand{\varnego}[0]{v^{-}}
\newcommand{\varbotho}[0]{v}
\newcommand{\varpos}[2]{v^{+}_{#1,#2}}
\newcommand{\varneg}[2]{v^{-}_{#1,#2}}
\newcommand{\varboth}[2]{v_{#1,#2}}
\newcommand{\constr}[2]{constr(#1,#2)}
\newcommand{\constrn}[3]{constr^{#3}(#1,#2)}

\newcommand{\call}[2]{\ensuremath{\text{\texttt{#1}}(#2)}}

\makeatletter
\newcommand\smallparagraph{\@startsection{paragraph}{4}{\z@}%
                       {4\p@ \@plus 0\p@ \@minus 0\p@}%
                       {-0.5em \@plus -0.22em \@minus -0.1em}%
                       {\normalfont\normalsize\bfseries}}
\makeatother

\usepackage[]{todonotes}

\title{Constraint-Based Monitoring of Hyperproperties\thanks{This work was partially supported by the German Research Foundation (DFG) as part of the Collaborative Research Center ``Methods and Tools for Understanding and Controlling Privacy'' (CRC 1223) and the Collaborative Research Center ``Foundations of Perspicuous Software Systems'' (CRC 248), and by the European Research Council (ERC) Grant OSARES (No. 683300).}}

\author{
Christopher Hahn \and
Marvin Stenger \and
Leander Tentrup
}

\authorrunning{C. Hahn \and M. Stenger \and L. Tentrup}

\institute{Reactive Systems Group, Saarland University\\\email{lastname@react.uni-saarland.de}}

\begin{document}

\maketitle

\begin{abstract}
Verifying hyperproperties at runtime is a challenging problem as hyperproperties, such as non-interference and observational determinism, relate multiple computation traces with each other. It is necessary to store previously seen traces, because every new incoming trace needs to be compatible with every run of the system observed so far. Furthermore, the new incoming trace poses requirements on \emph{future} traces. In our monitoring approach, we focus on those requirements by rewriting a hyperproperty in the temporal logic HyperLTL to a Boolean constraint system. A hyperproperty is then violated by multiple runs of the system if the constraint system becomes unsatisfiable. We compare our implementation, which utilizes either BDDs or a SAT solver to store and evaluate constraints, to the automata-based monitoring tool RVHyper.


\end{abstract}

\sloppy
\section{Introduction}
%\begin{itemize}
  %\item Monitoring of hyperproperties is an interesting and hard problem
  %\item Prior work on automata based: list problems: combinatoric, det. automaton, granularity of trace subsumption: Example for combinatoric+granularity
  %\item Rewriting-based successful for LTL
  %\item First, incomplete approach by Borzoo et al.
  %\item Contribution: Complete and provably correct for $\forall^2$, pick up example from before
%\end{itemize}
As today's complex and large-scale systems are usually far beyond the scope of classic verification techniques like model checking or theorem proving, we are in the need of light-weight monitors for controlling the flow of information.
%We are in the need for light-weight monitoring techniques for information-flow control in todays complex and large-scale systems as they are usually far beyond the scope of classic verification techniques like model checking or theorem proving.
By instrumenting efficient monitoring techniques in such systems that operate in an unpredictable privacy-critical environment, countermeasures will be enacted before irreparable information leaks happen.
Information-flow policies, however, cannot be monitored with standard runtime verification techniques as they relate \emph{multiple} runs of a system.
For example, \emph{observational determinism}~\cite{journals/jcs/McLean92,conf/sp/Roscoe95,conf/csfw/ZdancewicM03} is a policy stating that altering non-observable input has no impact on the observable behavior.
Hyperproperties~\cite{journals/jcs/ClarksonS10} are a generalization of trace properties and are thus capable of expressing information-flow policies.
%In runtime verification of Hyperproperties, which are \emph{sets of sets} of execution traces, we are concerned whether a \emph{set} of runs trough a system violates a given specification. This is in contrast to classic trace property verification where a single run suffices to determine a violation. Most prominent examples of hyperproperties are \emph{information-flow policies} like observational determinism, which states that for all pairs of runs the observable output should be the same until the observable input is not the same anymore, i.e., altering high security input has no impact on the observable output.
HyperLTL~\cite{conf/post/ClarksonFKMRS14} is a recently introduced temporal logic for hyperproperties, which extends Linear-time Temporal Logic (LTL)~\cite{conf/focs/Pnueli77} with trace variables and explicit trace quantification.
Observational determinism is expressed as the  formula $\forall \pi,\pi' \ldot (\mathit{out}_{\pi} \leftrightarrow \mathit{out}_{\pi'})\LTLweakuntil(\mathit{in}_{\pi} \nleftrightarrow \mathit{in}_{\pi'})$, stating that all traces $\pi,\pi'$ should agree on the output as long as they agree on the inputs.

%In this paper, we study runtime verification of Hyperproperties, expressed in HyperLTL.
In contrast to classic trace property monitoring, where a single run suffices to determine a violation, in runtime verification of HyperLTL formulas, we are concerned whether a \emph{set} of runs through a system violates a given specification.
In the common setting, those runs are given sequentially to the runtime monitor~\cite{conf/csfw/AgrawalB16,conf/rv/BonakdarpourF16,conf/rv/FinkbeinerHST17,conf/tacas/FinkbeinerHST18}, which determines if the given set of runs violates the specification.
%Runtime verification of HyperLTL have been studied extensively~\cite{}.
An alternative view on HyperLTL monitoring is that every new incoming trace poses requirements on future traces.
For example, the event $\set{\mathit{in}, \mathit{out}}$ in the observational determinism example above asserts that for every other trace, the output $\mathit{out}$ has to be enabled if $\mathit{in}$ is enabled.
Approaches based on static automata constructions~\cite{conf/csfw/AgrawalB16,conf/rv/FinkbeinerHST17,conf/tacas/FinkbeinerHST18} perform very well on this type of specifications, although their scalability is intrinsically limited by certain parameters:
The automaton construction becomes a bottleneck for more complex specifications, especially with respect to the number of atomic propositions. Furthermore, the computational workload grows steadily with the number of incoming traces, as every trace seen so far has to be checked against every new trace.
%For example~\todo{example}.
Even optimizations~\cite{conf/rv/FinkbeinerHST17}, which minimize the amount of traces that must be stored, turn out to be too coarse grained as the following example shows.
Consider the monitoring of the HyperLTL formula $\forall \pi, \pi'. \LTLsquare (a_\pi \rightarrow \neg b_{\pi'})$, which states that globally if $a$ occurs on any trace $\pi$, then $b$ is not allowed to hold on any trace $\pi'$, on the following incoming traces:
%\vspace{-5ex}
%{\scalebox{0.85}{\parbox{\linewidth}{
\begin{align}
&\tracebox{\hspace{1ex}\set{a}}\tracebox{\hspace{1.6ex}\{\}} \tracebox{\hspace{1.6ex}\{\}} \tracebox{\hspace{1.6ex}\{\}}  &&\textit{$\neg b$ is enforced on the 1st pos.} \label{intro-one}\\ 
&\tracebox{\hspace{1ex}\{a\}}\tracebox{\hspace{1ex}\set{a}} \tracebox{\hspace{1.6ex}\{\}} \tracebox{\hspace{1.6ex}\{\}}  &&\textit{$\neg b$ is enforced on the 1st and 2nd pos.} \label{intro-two}\\ 
&\tracebox{\hspace{1ex}\{a\}}\tracebox{\hspace{1.6ex}\set{}} \tracebox{\hspace{1ex}\set{a}} \tracebox{\hspace{1.6ex}\{\}}  &&\textit{$\neg b$ is enforced on the 1st and 3rd pos.}\label{intro-three}
\end{align}
%}}
In prior work~\cite{conf/rv/FinkbeinerHST17}, we observed that traces, which pose \emph{less requirements} on future traces, can safely be discarded from the monitoring process. In the example above, the requirements of trace~\ref{intro-one} are dominated by the requirements of trace~\ref{intro-two}, namely that $b$ is not allowed to hold on the first and second position of new incoming traces. Hence, trace~\ref{intro-one} must not longer be stored in order to detect a violation.
But with the proposed language inclusion check in~\cite{conf/rv/FinkbeinerHST17}, neither trace~\ref{intro-two} nor trace~\ref{intro-three} can be discarded, as they pose incomparable requirements.
They have, however, overlapping constraints, that is, they both enforce $\neg b$ in the first step.

To further improve the conciseness of the stored traces information, we use \emph{rewriting}, which is a more fine-grained monitoring approach.
The basic idea is to track the requirements that future traces have to fulfill, instead of storing a set of traces.
In the example above, we would track the requirement that $b$ is not allowed to hold on the first three positions of every freshly incoming trace.
%\todo{example optimiert}.
%This can be circumvented with sophisticated optimizations that store only the minimal set of traces needed for detecting a violation. However, storing
%Approaches based on a static automaton construction have the problem that a \emph{deterministic} automaton is needed for the monitoring process. The construction of such automata turns out to be infeasible quickly, for example, for a growing number of inputs in the information-flow policy. Even for policies for which monitoring templates can be constructed, the monitor has to store every trace seen so far in order to detect a violation. This can be circumvented with sophisticated optimizations that only store the minimal set of traces needed for detecting a violation. However, storing every full trace produces an enormous overhead during monitoring.
%~\todo{Example combinatoric+granularity}
%A runtime verification method that avoids this problem is \emph{rewriting},
Rewriting has been applied successfully to trace properties, namely LTL formulas~\cite{conf/kbse/HavelundR01}.
The idea is to partially evaluate a given LTL specification $\varphi$ on an incoming event by unrolling $\varphi$ according to the expansion laws of the temporal operators.
The result of a single rewrite is again an LTL formula representing the updated specification, which the continuing execution has to satisfy.
We use rewriting techniques to reduce $\forall^2$HyperLTL formulas to LTL constraints and check those constraints for inconsistencies corresponding to violations.
%The idea to use rewriting for HyperLTL monitoring was first considered in~\cite{conf/tacas/BrettSB17}, but their approach has major shortcomings (see related work).
%Monitoring HyperLTL formulas via rewriting has been considered before~\cite{conf/tacas/BrettSB17}.
%However, their approach is incomplete: The formula $\forall \pi_1,\pi_2,\pi_3 \ldot \LTLglobally (a_{\pi_1} \rightarrow (b_{\pi_2} \leftrightarrow b_{\pi_3})$ is still a hyperproperty after instantiating $\pi_1$, thus, it cannot be rewritten to LTL considering only one trace. It does not capture the full $\forall^2$ HyperLTL fragment considered here. Due to inconsistencies, their presented algorithm is not applicable to basic information flow policies, like observational determinism.

In this paper, we introduce a complete and provably correct rewriting-based monitoring approach for $\forall^2$HyperLTL formulas.
Our algorithm rewrites a HyperLTL formula and a single event into a constraint composed of plain LTL and HyperLTL.
For example, assume the event $\{\mathit{in},\mathit{out}\}$ while monitoring observational determinism formalized above.
The first step of the rewriting applies the expansion laws for the temporal operators, which results in $(in_{\pi} \nleftrightarrow in_{\pi'}) \lor (out_{\pi} \leftrightarrow out_{\pi'}) \land \LTLnext ((out_{\pi} \leftrightarrow out_{\pi'})\LTLweakuntil(in_{\pi} \nleftrightarrow in_{\pi'}))$.
The event \set{in, out} is rewritten for atomic propositions indexed by the trace variable $\pi$. This means replacing each occurrence of $in$ or $out$ in the current expansion step, i.e., before the $\LTLnext$ operator, with $\top$.
Additionally, we strip the $\pi'$ trace quantifier in the current expansion step from all other atomic propositions.
This leaves us with $(\top \nleftrightarrow in) \lor (\top \leftrightarrow out) \land \LTLnext ((out_{\pi} \leftrightarrow out_{\pi'})\LTLweakuntil(in_{\pi} \nleftrightarrow in_{\pi'}))$.
After simplification we have $\neg in \lor out \land \LTLnext ((out_{\pi} \leftrightarrow out_{\pi'})\LTLweakuntil(in_{\pi} \nleftrightarrow in_{\pi'}))$ as the new specification, which consists of a plain LTL part and a HyperLTL part.
Based on this, we incrementally build a Boolean constraint system: we start by encoding the constraints corresponding to the LTL part and encode the HyperLTL part as variables. Those variables will then be incrementally defined when more elements of the trace become available.
With this approach, we solely store the necessary information needed to detect violations of a given hyperproperty.

We evaluate two implementations of our approach, based on BDDs and SAT-solving, against RVHyper~\cite{conf/tacas/FinkbeinerHST18}, a highly optimized automaton-based monitoring tool for temporal hyperproperties. Our experiments show that the rewriting approach performs equally well in general and better on a class of formulas which we call \emph{guarded invariants}, i.e., formulas that define a certain invariant relation between two traces.

\smallparagraph{Related Work.}
%\section{Related Work}
%\label{sec:relatedwork}

\ifcomment
This chapter should give a broad overview on existing methods on monitoring
information flow policies, with the main focus on online monitoring.

\section{Static Analysis}
Sabelfeld et al.[24] survey the literature focusing on static program analysis
for enforcement of security policies. In some cases, with com- pilers using
Just-in-time compilation techniques and dynamic inclusion of code at run time
in web browsers, static analysis does not guarantee secure execution at run
time. Type systems, frameworks for JavaScript [6] and ML [21] are some
approaches to monitor information flow. Several tools [18, 11, 19] add
extensions such as statically checked information flow annotations to Java
language. Clark et al.[7] present verification of information flow for
deterministic interactive pro- grams. On the other hand, our approach is
capable of monitoring the subset of hyperproperties described by
alternation-free \hyperltl and not just informa- tion flow without assistance
from static analyzers. In [2], the authors propose a technique for designing
runtime monitors based abstract interpretation of the system under inspection.
\fi

\ifcomment
\section{Dynamic Analysis}
\fi

%In~\cite{DFR12}, they propose a monitoring approach using the
%temporal logic SecLTL~\cite{DFKRS12}, which extends LTL with a new modal operator \textit{Hide}, allowing to express
%non-interference as well as declassification.
%Making use of earlier work~\cite{FS04}, the authors base
%their monitoring algorithm on the on-the-fly construction of an accepting run trees for alternating automata.
%For a given system, represented as a transition system over the alphabet of
%inputs and outputs, and a specification expressed in SecLTL, their algorithm
%starts off construcing the monitor automaton. For each incoming trace the
%alternating automaton is traversed in a breath-first manner in order to
%construct a corresponding run tree.  When the algorithm determines that there
%is no valid possibility to construct a run tree, a violation to the
%specification is found.
%An explicit representation of the system is necessary for their method to be
%applicable.

%\todo{state background of hyperproperties}

With the need to express temporal hyperproperties in a succinct and formal manner,
%i.e., properties over sets of computation paths, new logics were proposed.
%The standard temporal logics such as LTL, CTL, and CTL* cannot express many
%hyperproperties of interest \ifcomment examples \fi, as they can only refer to
%a single computation path or tree, respectively, at a time.
the above mentioned temporal logics HyperLTL and HyperCTL*~\cite{conf/post/ClarksonFKMRS14} have been proposed.
The model-checking~\cite{conf/post/ClarksonFKMRS14,conf/cav/FinkbeinerRS15,conf/cav/FinkbeinerHT18}, satisfiability~\cite{conf/concur/FinkbeinerH16}, and realizability problem~\cite{conf/cav/FinkbeinerHLST18} of $\hyperltl$ has been studied before.

%Agrawal and Bonakdarpour~\cite{conf/csfw/AgrawalB16}, were the first to use HyperLTL for runtime verification of $k$-safety hyperproperties.
Runtime verification of HyperLTL formulas was first considered for (co-)$k$-safety hyperproperties~\cite{conf/csfw/AgrawalB16}.
%\todo{Tell something about Runtime Verification}
%A hyperproperty $\varphi$ is $k$-safety,
%if for each set of traces violating
%$\varphi$ a bad set of prefixes of those traces can be given with at most $k$
%elements.
%if for each violation of the property, the
%root cause for the violation is captured in at most $k$ finite prefixes of the
%computation.
%This class of hyperproperties still allows to express many relevant security
%policies~\cite{}.
%\todo{state policies somewhere}
%They show that HyperLTL formulas with at most $k$
%universal quantifiers can express a rich subset of $k$-safety hyperproperties.
In the same paper, the notion of monitorability for
HyperLTL was introduced.
%The concept of monitorability formalizes the intuitive idea of being able to
%detect acceptance or violation of a property by only observing a system at
%runtime. So the monitorable HyperLTL formulas are those for which it is either
%possible to find a witness, a set of finite prefixes, of acceptance or
%violation, respectively. 
The authors have also identified syntactic classes of
HyperLTL formulas that are monitorable and they proposed a monitoring algorithm based on a progression logic expressing trace interdependencies and the composition of an LTL$_3$ monitor.
%The monitoring algorithm they propose is based on a progression logic
%expressing trace inter-dependencies and the composition of LTL$_3$ monitor
%automata \cite{LTL3} in a Petri net whose special sink states represent
%verdicts for sub-formulas.
%Incoming traces get evaluated for each LTL$_3$ monitor in the Petri net and the progression logic keeps track of seen events for monitoring future traces.

Another automata-based approach for monitoring HyperLTL formulas was proposed in~\cite{conf/rv/FinkbeinerHST17}.
Given a HyperLTL specification, the algorithm starts by creating a deterministic monitor automaton.
For every incoming trace it is then checked that all combinations with the already seen traces are accepted by the automaton.
%By design the algorithm can not only detect a violation but also return a set of traces as a runtime verification verdict.
In order to minimize the number of stored traces, a language-inclusion-based algorithm is proposed, which allows to prune traces with redundant information.
Furthermore, a method to reduce the number of combination of traces which have to get checked by analyzing the specification for relations such as reflexivity, symmetry, and transitivity with a HyperLTL-SAT solver~\cite{conf/concur/FinkbeinerH16,conf/cav/FinkbeinerHS17}, is proposed.
The algorithm is implemented in the tool RVHyper~\cite{conf/tacas/FinkbeinerHST18}, which was used to monitor information-flow policies and to detect spurious dependencies in hardware designs.

%A first rewrite-based monitoring approach for alternation-free HyperLTL has been proposed in~\cite{conf/tacas/BrettSB17}.%~\todo{check phrasing, with bernd?}
%By analyzing the structure of the HyperLTL formula, they try to identify propositions of interest and store corresponding constraints, instead of whole traces. They aggregate a \emph{conjunctive} constraint system, such that violations to trace relations will manifest themselves in inconsistencies.
Another rewriting-based monitoring approach for HyperLTL is outlined in~\cite{conf/tacas/BrettSB17}. The idea is to identify a set of propositions of interest and aggregate constraints 
such that inconsistencies in the constraints indicate a violation of the HyperLTL formula. While the paper describes the building blocks for such a monitoring approach with a number of examples, we have, unfortunately, not been successful in applying the algorithm to other hyperproperties of interest, such as observational determinism.

In~\cite{conf/csfw/BonakdarpourF18}, the authors study the complexity of monitoring hyperproperties. They show that the form and size of the input, as well as the formula have a significant impact on the feasibility of the monitoring process. They differentiate between several input forms and study their complexity: a set of linear traces, tree-shaped Kripke structures, and acyclic Kripke structures. For acyclic structures and alternation-free HyperLTL formulas, the problems complexity gets as low as NC.

In~\cite{conf/isola/BonakdarpourSS18}, the authors discuss examples where static analysis can be combined with runtime verification techniques to monitor HyperLTL formulas beyond the alternation-free fragment. They discuss the challenges in monitoring formulas beyond this fragment and lay the foundations towards a general method.

%Runtime Enforcement of Security Policies on Black Box Reactive Programs~\todo{more}.

%just for inspiration, do not use as is
\ifcomment
Russo et al.[23] concentrate on permissive techniques for the enforcement of
information flow under flow-sensitivity. It has been shown that in the
flow-insensitive case, a sound purely dynamic monitor is more permissive than
static analysis. However, they show the impossibility of such a monitor in the
flow-sensitive case. A framework for inlining dynamic information flow monitors
has been presented by Magazinius et al.[14]. The approach by Chudnov et al.[5] uses
hybrid analysis instead and argues that due to JIT compilation processes, it is
no longer possible to mediate every data and control flow event of the native
code. They leverage the results of Russo et al.[23] by inlining the security
monitors. Chudnov et al.[4] again use hybrid analysis of 2-safety hyperproperties
in relational logic. In [1], the authors propose an automata-based RV technique
for monitoring only a disjunctive fragment of alternation-free \hyperltl.

Austin et al.[3] implement a purely dynamic monitor, however, restrictions such as
“no-sensitive upgrade” were placed. Some techniques deploy taint tracking and
labelling of data variables dynamically [26, 20]. Zdancewic et al.[25] verify
information flow for concurrent programs. Most of the techniques cited 14above
aim to monitor security policies described solely with two trace quantifiers
(without alternation), on observing a single run, whereas, our work is for any
hyperproperties that can be described with alternation-free \hyperltl, when
multiple runs are observed.
\fi

\ifcomment
\section{Secure Multi-Execution} 

Secure Multi-Execution [10] is a technique to enforce non-interference. In
SME, one executes a program multiple times, once for each security level, us-
ing special rules for I/O operations. Outputs are only produced in the execution
linked to their security level. Inputs are replaced by default inputs except in exe-
cutions linked to their security level or higher. Input side effects are supported by
making higher-security-level executions reuse inputs obtained in lower-security-
level threads. This approach is sound in a deterministic language.
While there are small similarities between SME and our work, there are fun-
damental differences. SME only focuses on non-interference and aims to enforce
it, but there are many critical hyperproperties that differ from non-interference
that our method is able to monitor. Thus, SME enforces a security policy at the
cost of restricting what it can enforce, whereas our technique monitors a much
larger set of policies.
\fi

\ifcomment
\section{Offline}

Offline Verification
Basin et al.[2] develop a model checker for security protocols. Since
traditional tools and verification methodologies are not equipped to deal with
sets of traces, several results introduce new logics or operators to express
hyperproperties.  \secltl extends \ltl by using an additional hide modality [3].
It allows expression of non-interference as well as the instance until a high
level data should remain independent of interference from low level data. The
modal $\mu$-calculus does not suffice to express some information flow properties.
Epistemic logic has been used to implicitly quantify over traces [18].
However, \hyperltl and \hyperltl$^{*}$ [4] subsume epistemic logic and quantified
propositional temporal logic [19].  In [5], the authors introduce a model
checking algorithm for verifying \hyperltl formulas.
\fi

%%% Local Variables:
%%% TeX-master: "thesis"
%%% End:

\section{Preliminaries}
Let $\mathit{AP}$ be a finite set of \emph{atomic propositions} and let $\Sigma = 2^\mathit{AP}$ be the corresponding \emph{alphabet}.
An infinite \emph{trace} $t \in  \Sigma^\omega$ is an infinite sequence over the alphabet.
A subset $T \subseteq \Sigma^\omega$ is called a \emph{trace property}. A \emph{hyperproperty} $H \subseteq 2^{(\Sigma^\omega)}$ is a generalization of a trace property.
A finite trace $t \in \Sigma^+$ is a finite sequence over $\Sigma$.
In the case of finite traces, $|t|$ denotes the length of a trace.
%Depending on the context, we refer to either $(2^\mathit{AP})^\omega$ or $(2^\mathit{AP})^+$ as the alphabet $\Sigma^\omega$ or $\Sigma^+$.
We use the following notation to access and manipulate traces:
Let $t$ be a trace and $i$ be a natural number.
$t[i]$ denotes the $i$-th element of $t$.
Therefore, $t[0]$ represents the first element of the trace.
Let $j$ be natural number. If $j \geq i$ and $i\geq|t|$, then
$t[i,j]$ denotes the sequence $t[i] t[i+1] \cdots t[min(j,|t|-1)]$. Otherwise it denotes the empty trace $\epsilon$.
$t[i\rangle$ denotes the suffix of $t$ starting at position $i$.
For two finite traces $s$ and $t$, we denote their concatenation by $s \cdot t$.

\smallparagraph{HyperLTL Syntax.}\label{thltl:syntax}
HyperLTL~\cite{conf/post/ClarksonFKMRS14} extends LTL with trace variables and trace quantifiers.
Let $\mathcal{V}$ be a finite set of trace variables.
The syntax of HyperLTL is given by the grammar
%\begin{figure}[h!]
\begin{alignat*}{3}
  \varphi  &{}\coloneqq ~ \forall \pi . \; \varphi \mid \exists \pi . \; \varphi \mid \psi  \\
  \psi &{}\coloneqq ~ a_{\pi} \mid \psi \wedge \psi \mid \neg \psi \mid \LTLnext \psi  \mid \psi \; \LTLuntil \;  \psi \enspace ,
\end{alignat*}
where $a \in \mathit{AP}$ is an atomic proposition and $\pi \in \mathcal{V}$ is a trace variable.
Atomic propositions are indexed by trace variables. The explicit trace quantification enables us to express properties like ``on all traces $\varphi$ must hold'', expressed by $\forall \pi \ldot \varphi$. Dually, we can express ``there exists a trace such that $\varphi$ holds'', expressed by $\exists \pi \ldot \varphi$.
We use the standard derived operators \emph{release} $\varphi \LTLrelease \psi \coloneqq \neg(\neg \varphi \LTLuntil \neg\psi)$, \emph{eventually} $\LTLdiamond \varphi \coloneqq \mathit{true} \LTLuntil \varphi$, \emph{globally} $\LTLsquare \varphi \coloneqq \neg \LTLdiamond \neg \varphi$, and \emph{weak until} $\varphi_1 \LTLweakuntil \varphi_2 \coloneqq (\varphi_1 \LTLuntil \varphi_2) \vee \LTLsquare \varphi_1$.
As we use the finite trace semantics, $\LTLnext \varphi$ denotes the \emph{strong} version of the next operator, i.e., if a trace ends before the satisfaction of $\varphi$ can be determined, the satisfaction relation, defined below, evaluates to false.
To enable duality in the finite trace setting, we additionally use the \emph{weak} next operator $\LTLweaknext \varphi$ which evaluates to true if a trace ends before the satisfaction of $\varphi$ can be determined and is defined as $\LTLweaknext \varphi \coloneqq \neg \LTLnext \neg \varphi$.
%$\LTLuntil$ denotes the \emph{until}-operator, $\LTLrelease$ the \emph{release}-operator, and we distinguish between a weak and a strong next operator, where $\LTLweaknext \varphi$ denotes the \emph{weak} version, i.e., if a trace ends before the satisfaction of $\varphi$ can be determined, the satisfaction relation, defined below, evaluates to true. In contrast, $\LTLnext \varphi$ denotes the \emph{strong} version of the next operator.
%$\LTLdiamond$ means that $\varphi$ holds eventually in the future, $\LTLsquare$ means that $\varphi$ holds globally, and $\LTLweakuntil$ is the \emph{weak} version of the \emph{until} operator.
%We call a HyperLTL formula alternation-free, if either all trace variables are universally quantified or all trace variables are existentially quantified.
%, and \emph{release} $\varphi_1 \LTLrelease \varphi_2 \equiv \neg (\neg\varphi_1 \LTLuntil \neg \varphi_2)$.
We call $\psi$ of a HyperLTL formula $\vec{Q}. \psi$, with an arbitrary quantifier prefix $\vec{Q}$, the \emph{body} of the formula. A HyperLTL formula $\vec{Q}. \psi$ is in the \emph{alternation-free fragment} if either $\vec{Q}$ consists solely of universal quantifiers or solely of existential quantifiers.
%We also call the $\forall^n$ fragment and $\exists^n$ fragment, with $n$ being a natural number, as the respective alternation-free fragments.
We also denote the respective alternation-free fragments as the $\forall^n$ fragment and the $\exists^n$ fragment, with $n$ being the number of quantifiers in the prefix.

\smallparagraph{Finite Trace Semantics.}
We recap the finite trace semantics for HyperLTL~\cite{conf/tacas/BrettSB17} which is itself based on the finite trace semantics of $\ltl$~\cite{books/daglib/0080029}.
% and the idea that temporal operators can only be satisfied if the finite traces are not yet terminated.
%We define a finite trace semantics for $\hyperltl$ based on the finite trace semantics of $\ltl$~\cite{books/daglib/0080029} and the idea that temporal operators can only be satisfied if the finite traces are not yet terminated.
In the following, when using $\lang{\varphi}$ we refer to the finite trace semantics of a HyperLTL formula $\varphi$.
%Let $t$ be a finite trace, $\epsilon$ denotes the empty trace, and $|t|$ denotes the length of a trace. Since we are in a finite trace setting, $t[i\rangle$ denotes the subsequence from position $i$ to position $|t|-1$.
%Since we are in a finite trace setting, $t[i\rangle$ denotes the subsequence from position $i$ to position $|t|-1$.
Let $\pathassignfin : \pathvars \rightarrow \Sigma^+$ be a partial function mapping trace variables to finite traces. We define $\epsilon[0]$ as the empty set.
$\pathassignfin[i\rangle$ denotes the trace assignment that is equal to $\pathassignfin(\pi)[i\rangle$ for all $\pi \in \dom(\pathassignfin)$.
By slight abuse of notation, we write $t \in \pathassignfin$ to access traces $t$ in the image of $\pathassignfin$.
The satisfaction of a HyperLTL formula $\varphi$ over a finite trace assignment $\pathassignfin$ and a set of finite traces $T$, denoted by $\pathassignfin \models_T \varphi$, is defined as follows:%~\todo{define min t in Pi fin}:
\begin{equation*}
\begin{array}{ll}
    \pathassignfin \models_T a_\pi         \qquad \qquad & \text{if } a \in \pathassignfin(\pi)[0] \\
    \pathassignfin \models_T \neg \varphi              & \text{if } \pathassignfin \nmodels_T \varphi \\
    \pathassignfin \models_T \varphi \lor \psi         & \text{if } \pathassignfin \models_T \varphi \text{ or } \pathassignfin \models_T \psi \\
    \pathassignfin \models_T \LTLnext \varphi          & \text{if } \forall t \in \pathassignfin \ldot |t|>1 \text{ and } \pathassignfin[1\rangle \models_T \varphi \\
%    \pathassignfin \models_T \LTLweaknext \varphi      & \text{if } \pathassignfin \models_T \LTLnext \varphi \\
%    & \text{or } \exists t \in \pathassignfin \text{ with } |t| \leq 1 \\
    \pathassignfin \models_T \varphi\LTLuntil\psi      & \text{if } \exists i < \min\nolimits_{t \in \pathassignfin} |t| \ldot \pathassignfin[i\rangle \models_T \psi \land \forall j < i \ldot \pathassignfin[j\rangle \models_T \varphi \\
%    \pathassignfin \models_T \varphi\LTLrelease\psi    & \text{if } \exists i < \min\nolimits_{t \in \pathassignfin} |t| \ldot \pathassignfin[i\rangle \models_T \varphi \land \forall j \leq i \ldot \pathassignfin[j\rangle \models_T \psi \\
%   & \text{or } \forall j \leq \min\nolimits_{t \in \pathassignfin} |t| \ldot \pathassignfin[j\rangle \models_T \varphi\\
    \pathassignfin \models_T \exists \pi \ldot \varphi & \text{if there is some } t \in T \text{ such that } \pathassignfin[\pi \mapsto t] \models_T \varphi \\
    \pathassignfin \models_T \forall \pi \ldot \varphi & \text{if for all } t \in T \text{ such that } \pathassignfin[\pi \mapsto t] \models_T \varphi
\end{array}
\end{equation*}
%\subsection{$\forall^2 \hyperltl$}
%As our constraint based monitoring approach will only be concerned with $\forall^2 \hyperltl$ formulas, we introduce special notation, which will help us down the road.
%\begin{definition}
%
Due to duality of $\LTLuntil$/$\LTLrelease$, $\LTLnext$/$\LTLweaknext$, $\exists$/$\forall$, and the standard Boolean operators, every HyperLTL formula $\varphi$ can be transformed into negation normal form~(NNF), i.e., for every $\varphi$ there is some $\psi$ in negation normal form such that for all $\pathassignfin$ and $T$ it holds that $\pathassignfin \models_T \varphi$ if, and only if, $\pathassignfin \models_T \psi$.
The standard LTL semantic, written $t \modelsltlfin \varphi$, for some LTL formula $\varphi$ is equal to $\set{ \pi \mapsto t }_\mathit{fin} \models_\emptyset \varphi'$, where $\varphi'$ is derived from $\varphi$ by replacing every proposition $p \in \mathit{AP}$ by $p_\pi$.

%======================================================================
\section{Rewriting HyperLTL}  \label{sec:rewriting}
%======================================================================
%We focus on $\forall^2\hyperltl$ HyperLTL formulas, i.e., formulas with exactly two universal quantifiers. \todo{add? a large subclass of hyperproperties fall into this category}
%
Given the body $\varphi$ of a $\forall^2$HyperLTL formula $\forall \pi,\pi' \ldot \varphi$, and a finite trace $t \in \Sigma^+$, we define alternative language characterizations.
These capture the intuitive idea that, if one fixes a finite trace $t$, the language of $\forall \pi,\pi' \ldot \varphi$ includes exactly those traces $t'$ that satisfy $\varphi$ in conjunction with $t$.
%we get a set of finite traces, which each in combination with $t$ satisfy $\varphi$:
\begin{equation*}
\begin{array}{ll}
\langxx{\varphi}{\pi}{t} & \defeq \left\{ t' \in \Sigma^+ | \left\{\pi \mapsto t, \pi' \mapsto t' \right\}_\mathit{fin} \models \varphi \right\} \\
\langxx{\varphi}{\pi'}{t} & \defeq \left\{ t' \in \Sigma^+ | \left\{\pi \mapsto t', \pi' \mapsto t \right\}_\mathit{fin} \models \varphi \right\} \\
\langx{\varphi}{t} & \defeq \langxx{\varphi}{\pi}{t} \cap \langxx{\varphi}{\pi'}{t}\\
%& ~ = \left\{ t' \in \Sigma^+ | \{t,t'\} \in \lang{\forall \pi,\pi' \ldot \varphi} \right\} \\
%\lang_t(\varphi) & = \left\{ t' \in \Sigma^+ | \left\{\pi \mapsto t, \pi' \mapsto t'\right\}_\mathit{fin} \models \varphi , \left\{\pi \mapsto t', \pi' \mapsto t \right\}_\mathit{fin} \models \varphi \right\} \\
\end{array}
\end{equation*}
%\end{definition}
We call $\hat{\varphi} \defeq \varphi \land \varphi[\pi'/\pi, \pi/\pi']$ the symmetric closure of $\varphi$, where $\varphi[\pi'/\pi, \pi/\pi']$ represents the expression $\varphi$ in which the trace variables $\pi, \pi'$ are swapped. %The symmetric closure has the following interesting property.
The language of the symmetric closure, when fixing one trace variable, is equivalent to the language of $\varphi$.
\begin{lemma}
	\label{lem:symm-closure}
	Given the body $\varphi$ of a $\forall^2$HyperLTL formula $\forall \pi,\pi' \ldot \varphi$, and a finite trace $t \in \Sigma^+$, it holds that
	%\begin{equation*}
	% \label{eq:symm1}
	$ \langxx{\hat{\varphi}}{\pi}{t} = \langx{\varphi}{t}$.
	%\end{equation*}
	%   \begin{equation}
	%       \label{eq:symm2}
	%       \hat{\hat{\varphi}} = \hat{\varphi}\\
	%   \end{equation}
\end{lemma}
\begin{proof}
	~\\
	%    \newline
	%    \begin{itemize}
	%        \item[(\ref{eq:symm1})]
	%\begin{equation*}
	$\langxx{\hat{\varphi}}{\pi}{t} = \left\{ t' \in \Sigma^+ | \left\{\pi \mapsto t, \pi' \mapsto t'\right\}_\mathit{fin} \models \hat{\varphi} \right\}$\\
	$\begin{array}{ll}
	&\hspace{5.9ex}= \left\{ t' \in \Sigma^+ | \left\{\pi \mapsto t, \pi' \mapsto t'\right\}_\mathit{fin} \models \varphi \land \varphi[\pi'/\pi, \pi/\pi'] \right\} \\
	&\hspace{5.9ex}= \left\{ t' \in \Sigma^+ | \left\{\pi \mapsto t, \pi' \mapsto t'\right\}_\mathit{fin} \models \varphi , \left\{\pi \mapsto t, \pi' \mapsto t' \right\}_\mathit{fin} \models \varphi[\pi'/\pi, \pi/\pi'] \right\} \\
	&\hspace{5.9ex}= \left\{ t' \in \Sigma^+ | \left\{\pi \mapsto t, \pi' \mapsto t'\right\}_\mathit{fin} \models \varphi , \left\{\pi \mapsto t', \pi' \mapsto t \right\}_\mathit{fin} \models \varphi \right\} 
	= \langx{\varphi}{t} 
	\end{array}$
	%\end{equation*}
	%        \item[(\ref{eq:symm2})]
	%    \end{itemize}
\end{proof}
%It is also easy to see that the following holds:\\
%\begin{equation*}
%    \forall \varphi, t \in \Sigma^+ \ldot \quad \langxx{\hat{\varphi}}{\pi}{t} = \langxx{\hat{\varphi}}{\pi'}{t}
%\end{equation*}
%\\
%\begin{definition}
%    We call the body $\varphi$ of a $\forall^2 \hyperltl$ formula $\forall \pi,\pi' \ldot \varphi$ to be symmetric w.r.t. finite semantics $\bf{symm}_\mathit{fin}$ if:\\
%\begin{equation*}
%    \forall t \in \Sigma^+ \ldot\newline
%    \langxx{\varphi}{\pi}{t} = \langxx{\varphi}{\pi'}{t}
%\end{equation*}
%\end{definition}
%\begin{lemma}
%\begin{equation*}
%   \varphi ~ \bf{symm}_\mathit{hyper} ~ \Rightarrow ~ \varphi ~ \bf{symm}_\mathit{fin}
%\end{equation*}
%\end{lemma}
%\begin{proof}
%   The finite semantics of $\hyperltl$ is consistent with the infinite semantics of $\hyperltl$ for the finite part of the traces.
%\end{proof}
%The language of a $\forall^2\hyperltl$ formula, where one trace variable is fixed, is preserved when constructing the symmetric closure defined above.
We exploit this to rewrite a $\forall^2$HyperLTL formula into an LTL formula.
%We define a projection of the formula and a finite trace down to standard LTL, such that the finite semantics are preserved.
%-------------------------------------------------------------------------------
%\smallparagraph{HyperLTL rewrite to $\ltl$.}
%-------------------------------------------------------------------------------
%We now have established, what the language of a $\forall^2 \hyperltl$ formula given a finite trace should be. The downside of this definition is however that it is rather unhandy for upcoming proof work. 
%So in this subsection, we define a way of projecting the formula and the finite trace down to standard $\ltl$, whereby the finite semantics get preserved. 
%\begin{definition}
    We define the projection $\proj{\varphi}{\pi}{t}$ of the body $\varphi$ of a $\forall^2$HyperLTL formula $\forall \pi,\pi' \ldot \varphi$ in NNF and a finite trace $t \in \Sigma^+$ to an $\ltl$ formula recursively on the structure of $\varphi$:
\begin{align*}
\begin{array}{lcllcl}
    \proj{a_\pi}{\pi}{t}         \qquad & \coloneqq &
        \begin{cases}
            \top & \text{if } a \in t[0]\\
            \bot & \text{otherwise}\\
        \end{cases} \qquad &
    \proj{\neg a_\pi}{\pi}{t}         \qquad & \coloneqq &
        \begin{cases}
            \top & \text{if } a \notin t[0]\\
            \bot & \text{otherwise}\\
        \end{cases}\\
    \proj{a_{\pi'}}{\pi}{t}             & \coloneqq & a &
    \proj{\neg a_{\pi'}}{\pi}{t}             & \coloneqq & \neg a\\
    %\proj{(\neg \varphi)}{\pi}{t}         & \coloneqq & \neg \proj{\varphi}{\pi}{t} \\
    \proj{(\varphi \lor \psi)}{\pi}{t}    & \coloneqq & \proj{\varphi}{\pi}{t} \lor \proj{\psi}{\pi}{t} &
    \proj{(\varphi \land \psi)}{\pi}{t}    & \coloneqq & \proj{\varphi}{\pi}{t} \land \proj{\psi}{\pi}{t}
\end{array}\\
\begin{array}{lcl}
    \proj{(\LTLnext \varphi)}{\pi}{t}     & \coloneqq &
        \begin{cases}
            \bot & \text{if } |t| \leq 1\\
            \LTLnext \proj{\varphi}{\pi}{t[1\rangle} & \text{otherwise}\\
        \end{cases}\\
    \proj{(\LTLweaknext \varphi)}{\pi}{t} & \coloneqq &
        \begin{cases}
            \top & \text{if } |t| \leq 1\\
            \LTLweaknext \proj{\varphi}{\pi}{t[1\rangle} & \text{otherwise}\\
        \end{cases}\\
    %F \models_P \tilde{\X} \varphi                & \text{if } F[1\rangle \models_P \varphi \\
    \proj{(\varphi \LTLuntil \psi)}{\pi}{t} & \coloneqq &
        \begin{cases}
            \proj{\psi}{\pi}{t} & \text{if } |t| \leq 1\\
            \proj{\psi}{\pi}{t} \lor (\proj{\varphi}{\pi}{t} \land \LTLnext (\proj{(\varphi \LTLuntil \psi)}{\pi}{t[1\rangle})) & \text{otherwise}\\
        \end{cases}\\
    \proj{(\varphi \LTLrelease \psi)}{\pi}{t} & \coloneqq &
        \begin{cases}
            \proj{\psi}{\pi}{t} & \text{if } |t| \leq 1\\
            \proj{\psi}{\pi}{t} \land (\proj{\varphi}{\pi}{t} \lor \LTLweaknext (\proj{(\varphi \LTLrelease \psi)}{\pi}{t[1\rangle})) & \text{otherwise}\\
        \end{cases}\\
\end{array}
\end{align*}
%\end{definition}

\begin{theorem}
\label{thm:ltlfin-compat}
%For any body $\varphi$ of a $\forall^2 \hyperltl$ formula $\forall \pi,\pi' \ldot \varphi$ and any finite trace $t \in \Sigma^+$ it holds:\\
Given a $\forall^2$HyperLTL formula $\forall \pi, \pi' \ldot \varphi$ and any two finite traces $t,t' \in \Sigma^+$ it holds that
    %\begin{equation*}
        $t' \in \langxx{\varphi}{\pi}{t}
        %\xLeftrightarrow{\text{def}} \left\{\pi \mapsto t, \pi' \mapsto t'\right\}_\mathit{fin} \models \varphi
        \text{ if, and only if } t' \modelsltlfin \proj{\varphi}{\pi}{t}$.
    %\end{equation*}
\end{theorem}
\begin{proof}
    By induction on the size of $t$.
    %\begin{itemize}
    %	\item
    Induction Base ($t=e$, where $e \in \Sigma$): Let $t' \in \Sigma^+$ be arbitrarily chosen.
    %In order to proof the base of the induction, we have to do a case distinction over the formula  $\varphi$ in the fashion of structural induction:
    %Structural Induction Base Cases:\\
    We distinguish by structural induction the following cases over the formula $\varphi$. We begin with the base cases.
    %\vspace{-2ex}
    \begin{itemize}
        \item$a_{\pi}$:
            we know by definition that  $\proj{a_{\pi}}{\pi}{t}$ equals $\top$ if $a \in t[0]$ and $\bot$ otherwise,
            so it follows that
            $t' \modelsltlfin \proj{a_{\pi}}{\pi}{t} \Leftrightarrow a \in t[0] \Leftrightarrow t' \in \langxx{a_{\pi}}{\pi}{t}$.
        \item
        $a_{\pi'}$:
            $
            %\begin{array}{ll}
                t' \in \langxx{a_{\pi'}}{\pi}{t}
                \Leftrightarrow  a \in t'[0]
                \Leftrightarrow  t' \modelsltlfin a
                \Leftrightarrow  t' \modelsltlfin \proj{a_{\pi'}}{\pi}{t}.
            %\end{array}
            $
        \item 
            $\neg a_{\pi}$ and $\neg a_{\pi'}$ are proven analogously.
    \end{itemize}
%\vspace{-2ex}
    %We proceed with the structural Induction Step:\\
    The structural induction hypothesis states that
    $\forall t' \in \Sigma^+ \ldot t' \in \langxx{\psi}{\pi}{t} \Leftrightarrow t' \modelsltlfin \proj{\psi}{\pi}{t}$ (SIH1)
    %holds for all strict subformulas $\psi$.
    , where $\psi$ is a strict subformula of $\varphi$.
    %\vspace{-2ex}
    \begin{itemize}
%        \item$\neg \varphi$:
%            $
%            %\begin{array}{ll}
%                 t' \in \langxx{\neg \varphi}{\pi}{t}
%                \Leftrightarrow  t' \notin \langxx{\varphi}{\pi}{t}
%                \xLeftrightarrow{\text{SIH1}}  t' \nmodelsltlfin \proj{\varphi}{\pi}{t}
%                \Leftrightarrow  t' \modelsltlfin \neg \proj{\varphi}{\pi}{t}
%                \Leftrightarrow  t' \modelsltlfin \proj{(\neg \varphi)}{\pi}{t}.
%            %\end{array}
%            $
        \item$\varphi \lor \psi$:
            $
            %\begin{array}{ll}
                 t' \in \langxx{\varphi \lor \psi}{\pi}{t}
                \Leftrightarrow  (t' \in \langxx{\varphi}{\pi}{t}) \lor (t' \in \langxx{\psi}{\pi}{t})
                \xLeftrightarrow{\text{SIH1}}  (t' \modelsltlfin \proj{\varphi}{\pi}{t}) \lor (t' \modelsltlfin \proj{\psi}{\pi}{t})
                \Leftrightarrow  t' \modelsltlfin \proj{(\varphi \lor \psi)}{\pi}{t}.
            %\end{array}
            $
        \item$\LTLnext \varphi$:
            $t' \in \langxx{\LTLnext \varphi}{\pi}{t} \xLeftrightarrow{|t|=1} \bot \xLeftrightarrow{|t|=1} t' \modelsltlfin \proj{(\LTLnext \varphi)}{\pi}{t}$.
        %\item
       % $\LTLweaknext \varphi$:
         %   $t' \in \langxx{\LTLweaknext \varphi}{\pi}{t} \xLeftrightarrow{|t|=1} \top \xLeftrightarrow{|t|=1} t' \modelsltlfin \proj{(\LTLweaknext \varphi)}{\pi}{t}$.\\
        \item$\varphi \LTLuntil \psi$:
            $
           % \begin{array}{ll}
                 t' \in \langxx{\varphi \LTLuntil \psi}{\pi}{t}
                \xLeftrightarrow{|t|=1}  t' \in \langxx{\psi}{\pi}{t}
                \xLeftrightarrow{\text{SIH1}}  t' \modelsltlfin \proj{\psi}{\pi}{t}
                \xLeftrightarrow{|t|=1}  t' \modelsltlfin \proj{(\varphi \LTLuntil \psi)}{\pi}{t}.
            %\end{array}
            $
        \item $\varphi \wedge \psi$, $\LTLweaknext \varphi$ and $\varphi \LTLrelease \psi$ are proven analogously.
    \end{itemize}
	%\item
	%\vspace{-2ex}
    Induction Step ($t = e \cdot t^*$, where $e \in \Sigma, t^* \in \Sigma^+$):
    The induction hypothesis states that $\forall t' \in \Sigma^+ \ldot t' \in \langxx{\varphi}{\pi}{t^*} \Leftrightarrow t' \modelsltlfin \proj{\varphi}{\pi}{t^*}$ (IH).
    %It is left to show $\forall e \in \Sigma \text{ with } t^* \defeq e \cdot t \ldot \forall t' \in \Sigma^+ \ldot t' \in \langxx{\varphi}{\pi}{t^*} \Leftrightarrow t' \modelsltlfin \proj{\varphi}{\pi}{t^*}$.
    %Choosing $e \in \Sigma$ and therefore $t^* = e \cdot t$, as well as $t' \in \Sigma^+$, arbitrarily we have to show that
    %$t' \in \langxx{\varphi}{\pi}{t^*} \Leftrightarrow t' \modelsltlfin \proj{\varphi}{\pi}{t^*}$.
    We make use of structural induction over $\varphi$. All cases without temporal operators are covered as their proofs above were independent of $|t|$.
    %This leaves us with the proof obligation for $\LTLnext \varphi$ and $\varphi \LTLuntil \psi$.
    The structural induction hypothesis states for all strict subformulas $\psi$ that
    %We also use the structural induction hypothesis:
    $\forall t' \in \Sigma^+ \ldot t' \in \langxx{\psi}{\pi}{t} \Leftrightarrow t' \modelsltlfin \proj{\psi}{\pi}{t}$ (SIH2).
    %for strict subexpressions $\psi$.
    %where $\psi$ is a strict subformula.
    %Structural Induction Step:\\
    %\vspace{-2ex}
    \begin{itemize}
        \item$\LTLnext \varphi$:
            $
            %\begin{array}{ll}
                 t' \in \langxx{\LTLnext \varphi}{\pi}{t^*}
                \xLeftrightarrow{|t^*| \geq 2}  t'[1\rangle \in \langxx{\varphi}{\pi}{t}
                \xLeftrightarrow{\text{IH}}  t'[1\rangle \modelsltlfin \proj{\varphi}{\pi}{t}
                \Leftrightarrow  t' \modelsltlfin \LTLnext (\proj{\varphi}{\pi}{t})
                \xLeftrightarrow{t^* = e \cdot t}  t' \modelsltlfin \proj{(\LTLnext \varphi)}{\pi}{t^*}.
            %\end{array}
            $
        %\item$\LTLweaknext \varphi$: analogously
        \item$\varphi \LTLuntil \psi$:
            $
            %\begin{array}{ll}
                 t' \in \langxx{\varphi \LTLuntil \psi}{\pi}{t^*}
                \xLeftrightarrow{|t^*| \geq 2}  (t' \in \langxx{\psi}{\pi}{t^*}) \lor (t' \in \langxx{\varphi}{\pi}{t^*}) \land (t'[1\rangle \in \langxx{\varphi \LTLuntil \psi}{\pi}{t})
                \xLeftrightarrow{\text{SIH2+IH}}  (t' \modelsltlfin \proj{\psi}{\pi}{t^*}) \lor (t' \models \proj{\varphi}{\pi}{t^*}) \land (t'[1\rangle \modelsltlfin \proj{(\varphi \LTLuntil \psi)}{\pi}{t})
                \Leftrightarrow (t' \modelsltlfin \proj{\psi}{\pi}{t^*}) \lor (t' \models \proj{\varphi}{\pi}{t^*}) \land (t' \modelsltlfin \LTLnext(\proj{(\varphi \LTLuntil \psi)}{\pi}{t}))
                \Leftrightarrow  t' \modelsltlfin \proj{(\varphi \LTLuntil \psi)}{\pi}{t^*}.
            %\end{array}
            $
        \item $\LTLweaknext \varphi$ and $\varphi \LTLrelease \psi$ are proven analogously.
    \end{itemize}
	%\end{itemize}
	%\qed
\end{proof}

%The rewriting approach introduced in this section is the base for our monitoring approach.
%Given some trace $t$, $\langxx{\varphi}{\pi}{t}$ represents those traces that satisfy $\varphi$ in conjunction with $t$.
%Thus, $\langxx{\varphi}{\pi}{t}$ represents \emph{requirements} on future traces.
%If the conjunction of these requirements become unsatisfiable, we found a monitoring violation.

\section{Constraint-based Monitoring}

%\subsection{formal argument}

%\subsection{Algorithm}
%In \autoref{sec:rewriting}, we have seen how to rewrite a $\forall \pi, \pi' \ldot \varphi$ HyperLTL formula into LTL given a finite trace $t$ using $\proj{\varphi}{\pi}{t}$.
%For monitoring, this trace is given as a sequence of events and we want to detect violations as early as possible, especially we do not want to wait for the trace to end.
For monitoring, we need to define an \emph{incremental} rewriting that accurately models the semantics of $\proj{\varphi}{\pi}{t}$ while still being able to detect violations early.
To this end, we define an operation $\varphi[\pi,e,i]$, where $e \in \Sigma$ is an event and $i$ is the current position in the trace.
$\varphi[\pi,e,i]$ transforms $\varphi$ into a propositional formula, where the variables are either indexed atomic propositions $p_i$ for $p \in \AP$, or a variable $\varneg{\varphi'}{i+1}$ and $\varpos{\varphi'}{i+1}$ that act as placeholders until new information about the trace comes in.
Whenever the next event $e'$ occurs, the variables are defined with the result of $\varphi'[\pi,e',i+1]$.
If the trace ends, the variables are set to \textit{true} and \textit{false} for $\varposo$ and $\varnego$, respectively.
We define $\varphi[\pi,e,i]$ of a $\forall^2$HyperLTL formula $\forall \pi,\pi' \ldot \varphi$ in NNF, event $e \in \Sigma$, and $i \geq 0$ recursively on the structure of the body $\varphi$:
\begin{align*}
\begin{array}{lcllcl}
a_\pi[\pi,e,i]         \qquad & \coloneqq &
\begin{cases}
\top & \text{if } a \in e\\
\bot & \text{otherwise}\\
\end{cases} \qquad &
(\neg a_\pi)[\pi,e,i]         \qquad & \coloneqq &
\begin{cases}
\top & \text{if } a \notin e\\
\bot & \text{otherwise}\\
\end{cases}\\
a_{\pi'}[\pi,e,i]             & \coloneqq & a_i &
(\neg a_{\pi'})[\pi,e,i]             & \coloneqq & \neg a_i\\
%(\neg \varphi)[\pi,e,i]         & \coloneqq & \neg \varphi[\pi,e,i]\\
(\varphi \lor \psi)[\pi,e,i]    & \coloneqq & \varphi[\pi,e,i] \lor \psi[\pi,e,i] &
(\varphi \land \psi)[\pi,e,i]    & \coloneqq & \varphi[\pi,e,i] \land \psi[\pi,e,i]\\
(\LTLnext \varphi)[\pi,e,i]     & \coloneqq & \varneg{\varphi}{i+1}&
(\LTLweaknext \varphi)[\pi,e,i] & \coloneqq & \varpos{\varphi}{i+1}
\end{array}\\[\medskipamount]
\begin{array}{lcl}
(\varphi\LTLuntil\psi)[\pi,e,i] & \coloneqq & \psi[\pi,e,i] \lor (\varphi[\pi,e,i] \land \varneg{\varphi \LTLuntil\psi}{i+1})\\
(\varphi\LTLrelease\psi)[\pi,e,i] & \coloneqq & \psi[\pi,e,i] \land (\varphi[\pi,e,i] \lor \varpos{\varphi \LTLrelease\psi}{i+1}) \hspace{2cm}\mbox{}
\end{array}
\end{align*}
We encode a $\forall^2$HyperLTL formula and finite traces into a constraint system, which, as we will show, is satisfiable if and only if the given traces satisfy the formula w.r.t. the finite semantics of HyperLTL.
We write $\varboth{\varphi}{i}$ to denote either $\varneg{\varphi}{i}$ or $\varpos{\varphi}{i}$.
%\todo{add intuition behind the implication: if s-variable has to hold in order to achieve satisfication, then also it's obligations have to hold}
%\begin{definition}
For $e \in \Sigma$ and $t \in \Sigma^*$, we define
\begin{equation*}
\begin{array}{ll}
\constr{\varpos{\varphi}{i}}{\epsilon}   & \defeq \top \\
\constr{\varneg{\varphi}{i}}{\epsilon}       & \defeq \bot \\
\constr{\varboth{\varphi}{i}}{e \cdot t} & \defeq \varphi[\pi, e, i] \land \bigwedge\limits_{\varboth{\psi}{i+1} \in \varphi[\pi, e, i]} \Big( \varboth{\psi}{i+1} \imp \constr{\varboth{\psi}{i+1}}{t} \Big) \\
\encnx{\ap}{\epsilon}{i} & \defeq \top\\
\encnx{\ap}{e \cdot t}{i} & \defeq \bigwedge\limits_{a \in \ap\cap e} a_i \quad \land \bigwedge\limits_{a \in \ap\setminus e} \neg a_i \quad \land \quad \encnx{\ap}{t}{i+1} \enspace ,
\end{array}
\end{equation*}
where we use $\varboth{\psi}{i+1} \in \varphi[\pi, e, i]$ to denote variables $\varboth{\psi}{i+1}$ occurring in the propositional formula $\varphi[\pi, e, i]$.
$\mathit{enc}$ is used to transform a trace into a propositional formula, e.g., $\encnx{\set{a,b}}{\set{a}\set{a,b}}{0} \allowbreak = a_0 \land \neg b_0 \land a_1 \land b_1$.
For $n=0$ we omit the annotation, i.e., we write $\encx{\ap}{t}$ instead of $\encnx{\ap}{t}{0}$.
Also we omit the index $\ap$ if it is clear from the context.
By slight abuse of notation, we use $\constrn{\varphi}{t}{n}$ for some quantifier free HyperLTL formula $\varphi$ to denote $\constr{\varboth{\varphi}{n}}{t}$ if $\card{t} > 0$.
%\end{definition}
%\begin{definition}
For a trace $t' \in \Sigma^+$, we use the notation
%\begin{equation*}
%\begin{array}{ll}
$\enc{t'} \models \constr{\varphi}{t} \text{, which evaluates to \textit{true} if, and only if } \enc{t'} \land \constr{\varphi}{t} \text{ is satisfiable}.$
%\enc{t'} \nmodels \constr{\varphi}{t} \quad \quad & \text{if } \enc{t'} \land \constr{\varphi}{t} \text{ is \textbf{not satisfiable}.}\\
%\end{array}
%\end{equation*}
%\end{definition}
\subsection{Algorithm}
\label{alg}
Figure~\ref{fig:algorithms-constraint-based} depicts our constraint-based algorithm. Note that this algorithm can be used in an offline and online fashion.
Before we give algorithmic details, consider again, the observational determinism example from the introduction, which is expressed as $\forall^2$HyperLTL formula $\forall \pi,\pi' \ldot (out_{\pi} \leftrightarrow out_{\pi'})\LTLweakuntil(in_{\pi} \nleftrightarrow in_{\pi'})$.
The basic idea of the algorithm is to transform the HyperLTL formula to a formula consisting partially of LTL, which expresses the requirements of the incoming trace in the current step, and partially of HyperLTL.
Assuming the event \set{in,out}, we transform the observational determinism formula to the following formula: $\neg in \lor out \land \LTLnext ((out_{\pi} \leftrightarrow out_{\pi'})\LTLweakuntil(in_{\pi} \nleftrightarrow in_{\pi'}))$.
%We show how this formula is transformed if given the event \set{in,out}.
%As we can assume the quantifier structure of the HyperLTL specification to be fixed to two, we do not take care of this section of the specification and basically ignore it for most part.
%We strip of the quantifiers and the first step of the rewriting applies the expansion laws for the temporal operators.
%This leaves us with $(in_{\pi} \nleftrightarrow in_{\pi'}) \lor (out_{\pi} \leftrightarrow out_{\pi'}) \land \LTLweaknext ((out_{\pi} \leftrightarrow out_{\pi'})\LTLweakuntil(in_{\pi} \nleftrightarrow in_{\pi'}))$.
%Now rewrite the event \set{in,out} for atomic propositions with trace quantifier $\pi$. This means replacing each occurence of $in$ or $out$ in the current expansion step, i.e., before the $\LTLnext$ operator, with $\top$.
%Additionally, we strip the $\pi'$ trace quantifier in the current expansion step from all other atomic propositions.
%This leaves us with $(\top \nleftrightarrow in) \lor (\top \leftrightarrow out) \land \LTLnext ((out_{\pi} \leftrightarrow out_{\pi'})\LTLweakuntil(in_{\pi} \nleftrightarrow in_{\pi'}))$.
%After simplification we get $\neg in \lor out \land \LTLnext ((out_{\pi} \leftrightarrow out_{\pi'})\LTLweakuntil(in_{\pi} \nleftrightarrow in_{\pi'}))$ as the new specification, which consists, as mentioned above, partially of a plain LTL part and a HyperLTL part.
\begin{wrapfigure}{l}{0.52\textwidth}
  %\vspace{-15pt}
    \begin{algorithm}[H]
    \label{alg:algorithms-constraint-based1}
        \SetKwInOut{Input}{Input}
        \SetKwInOut{Output}{Output}
        %\SetAlgoLined
        %\KwResult{Write here the result }
        \Input{$\forall \pi,\pi'\ldot \varphi$, $T \subseteq \Sigma^+$}
        \Output{\textit{violation} or \textit{no violation}}
        \BlankLine
        $\psi \defeq \call{nnf}{\hat{\varphi}}$\\\label{line:transform-formula1}
        $C \defeq \top$\\\label{line:rootconstraint1}
        \ForEach{$t \in T$}{
            $C_t \defeq \varboth{\psi}{0}$\\\label{line:traceconstraint1}
            $t_\mathit{enc} \defeq \top$\\
            \While{$e_i \defeq \call{getNextEvent}{t}$}{
                $t_\mathit{enc} \defeq t_\mathit{enc} \land \encn{e_i}{i}$\\\label{line:enc-construct1}
                \ForEach{$\varboth{\phi}{i} \in C_t$}{
                    $c \defeq \phi[\pi,e_i,i]$\\\label{line:rewrite1}
                    %$\Phi \defeq \call{extract}{c}$\\\label{line:extract1}
                    $C_t \defeq C_t \land (\varboth{\phi}{i} \imp c)$\\\label{line:add-constr1}
                }
                \If{$\neg \call{sat}{C \land C_t \land t_\mathit{enc}}$}{
                    \Return \textit{violation}
                }
            }
            \ForEach{$\varpos{\phi}{i+1} \in C_t$}{
                $C_t \defeq C_t \land \varpos{\phi}{i+1}$\\
            }
            \ForEach{$\varneg{\phi}{i+1} \in C_t$}{
                $C_t \defeq C_t \land \neg \varneg{\phi}{i+1}$\\
            }
            $C \defeq C \land C_t$
        }
        \Return \textit{no violation}
        \BlankLine
        %\caption{Online Algorithm.}
    \end{algorithm}
    \caption{Constraint-based algorithm for monitoring $\forall^2\hyperltl$ formulas.}
    \label{fig:algorithms-constraint-based1}
    \label{fig:algorithms-constraint-based}
  \vspace{-17pt}
\end{wrapfigure}
A Boolean constraint system is then build incrementally: we start encoding the constraints corresponding to the LTL part (in front of the next-operator) and encode the HyperLTL part (after the next-operator) as variables that are defined when more events of the trace come in.
%
% which directly corresponds to the construction we have seen in \autoref{sec:constraint-based-proof}.
%The formal proof can be found in \autoref{sec:constraint-based-proof}.
%\emph{Unoptimized Algorithm.}
%Input is the $\forall^2$ HyperLTL formula $\forall \pi,\pi' \ldot \varphi$.
%this is the formal representation of the specification incoming traces of events will get monitored against. 
%Initially we have to transform the given formula such that our algorithm applies for the established correctness arguments (see \autoref{coro:set-same-lang-hyper-sat}).
We continue by explaining the algorithm in detail.
In \autoref{line:transform-formula1}, we construct $\psi$ as the negation normal form of the symmetric closure of the original formula.
We build two constraint systems: $C$ containing constraints of previous traces and $C_t$ (built incrementally) containing the constraints for the current trace $t$.
Consequently, we initialize $C$ with $\top$ and $C_t$ with $\varboth{\psi}{0}$ (lines~\ref{line:rootconstraint1} and \ref{line:traceconstraint1}).
If the trace ends, we define the remaining $\varbotho$ variables according to their polarities and add $C_t$ to $C$.
%We initialize the constraint system $C$ in \autoref{line:rootconstraint1} with the root constraint $s_{\psi,0}$, thereby $\varphi$ is not violated as long as the constraint system is satisfiable.
%Both $\psi$ and $s_{\psi,0}$ also reoccur in the initial assignment for formula set $F$, which contains pairs of formulas and their associated variables representing them in the constraint system.
%Each formula in $F$ corresponds to an obligation left to be satisfied by the next event.
For each new event $e_i$ in the trace $t$, and each ``open'' constraint in $C_t$ corresponding to step $i$, i.e., $\varboth{\phi}{i} \in C_t$, we rewrite the formula $\phi$ (line~\ref{line:rewrite1}) and define $\varboth{\phi}{i}$ with the rewriting result, which, potentially introduced new open constraints $\varboth{\phi'}{i+1}$ for the next step $i+1$.
%The implementation of this idea is in \autoref{line:add-constr1}, where for a formula $\phi$ and it's rewrite result $c$ the implication $s_{\phi,j} \imp c$ is added to the constraint system $C$. The intuition here is that constraint $c$ only has to be satisfied if it is necessary for $s_{\phi,j}$ to hold in order to satisfy the constraint system.
The constraint encoding of the current trace is aggregated in constraint $t_\mathit{enc}$ (\autoref{line:enc-construct1}).
%The name $A$ is chosen as an abbreviation for assumptions, which is the term used for the way we use $A$ for checking satisfiability of $C$.
If the constraint system given the encoding of the current trace turns out to be unsatisfiable, a violation to the specification is detected, which is then returned.

In the following, we sketch two algorithmic improvements.
First, instead of storing the constraints corresponding to traces individually, we use a new data structure, which is a \emph{tree maintaining nodes} of formulas, their corresponding variables and also child nodes.
Such a node corresponds to already seen rewrites.
%Instead of the formula set $F$, which for each incoming trace to captures subformulas left to be satisfied, we use a node set $N$, replacing $F$.
%Nodes contain a formula, its corresponding variable and also child nodes which again correspond to already seen rewrites of the formula.
The initial node captures the (transformed) specification (similar to line~\ref{line:traceconstraint1}) and it is also the root of the tree structure, representing all the generated constraints which replaces $C$ in Fig.~\ref{fig:algorithms-constraint-based}.
Whenever a trace deviates in its rewrite result a new child or branch is added to the tree.
%Although similar to the trie optimization in \todo{link tree algorithm}, we note that this solution is in a sense superior.
%While the trie optimization is only sensible for event equality (modulo actually used atomic propositions), looking at rewrite results allows to compare the events by their actual effect\todo{what effect?}.
%It may well be that there are two distinct events leading to the same rewrite result.
%A simple example for this situation would be the formula $(a_\pi \lor b_\pi) \imp c_\pi'$ paired with events \set{a}, \set{b}.
%For both events rewriting would yield the constraint $c$.
%So they impose the same constraints on future traces and therefore are effectively the same.
%The trie optimization cannot be aware of this.\\
%By the previous description of the use of nodes for the improved algorithm one advantage became clear.
If a rewrite result is already present in the node tree structure there is no need to create any new constraints nor new variables.
This is crucial in case we observe many equal traces or traces behaving effectively the same.
%Same argument applies here as in the automata-based method.
%Without being sensible for such trace redundancies we could not scale well in the number of traces, as we would run out of resources quickly, holding true for computational resources as well as for memory capacity.
In case no new constraints were added to the constraint system, we omit a superfluous check for satisfiability.
%as it has to be checked already at some earlier point during the monitoring session.
%Since the satisfiability check is the single most expensive step of our algorithm when it comes to complexity and computation time, \todo{fix sentence}.
%\begin{itemize}
%    \item skip solve if only present rewrites
%\end{itemize}
%\emph{Conjunct Splitting.}
%\textbf{Syntactic Splitting.}
%With the node tree structure presented above we have seen that there are substantial performance benefits to be made storing already seen rewrites.

Second, we use \emph{conjunct splitting} to utilize the node tree optimization even more.
%If we take a look at the syntactic structure of the constraints resulting from the rewrites, we could utilize the node tree optimization even more.
%
%\todo{motivate conjunct splitting}
We illustrate the basic idea on an example.
Consider $\forall \pi, \pi' \ldot \varphi$ with $\varphi = \LTLglobally ((a_\pi \leftrightarrow a_\pi') \lor (b_\pi \leftrightarrow b_\pi'))$, which demands that on all executions on each position at least on of propositions $a$ or $b$ agree in its evaluation.
Consider the two traces $t_1 = \set{a}\set{a}\set{a}$, $t_2 = \set{a}\set{a,b}\set{a}$ that satisfy the specification.
As both traces feature the same first event, they also share the same rewrite result for the first position.
Interestingly, on the second position, we get $(a \lor \neg b) \land s_\varphi$ for $t_1$ and $(a \lor b) \land s_\varphi$ for $t_2$ as the rewrite results.
While these constraints are no longer equal, by the nature of invariants, both feature the same subterm on the right hand side of the conjunction.
%, i.e., it also captures the invariant left to be satisfied for the next position.
%We would in this case create two different nodes for each rewrite, introducing a branch in the node tree.
%\todo{mention possible negative implications}
%\\
We split the resulting constraint on its syntactic structure, such that we would no longer have to introduce a branch in the tree.
%If the rewrite result is a conjunction, we handle each of its conjuncts as if it was the sole result of a rewrite.
%Instead of satisfying a single conjunct, we now require satisfaction of each conjunct, which is semantically equivalent.
%Applied on the above example, we would no longer have to introduce a branch in the node tree.
%So we can further profit from the node optimization for \todo{fix sentence}

%\todo{TODO}
%\begin{itemize}
%        \item share nodes/variables even if rewrites differs
%        \item reuse the variable of the parent node
%\end{itemize}

%\section{Constraint Based Monitoring}
\subsection{Correctness}
\label{sec:constraint-based-proof}
In this technical subsection, we will formally prove correctness of our algorithm by showing that our incremental construction of the Boolean constraints is equisatisfiable to the HyperLTL rewriting presented in Section~\ref{sec:rewriting}.
We begin by showing that satisfiability is preserved when shifting the indices, as stated by the following lemma.
\begin{lemma}
\label{lem:encshift}
    For any $\forall^2$HyperLTL formula $\forall \pi,\pi' \ldot \varphi$ over atomic propositions $\ap$, any finite traces $t,t' \in \Sigma^+$ and $n \geq 0$ it holds that
    %\begin{equation*}
    $
        \encx{\ap}{t'} \models \constr{\varphi}{t} \Leftrightarrow \encnx{\ap}{t'}{n} \models \constrn{\varphi}{t}{n}$.
    %\end{equation*}
\end{lemma}
\begin{proof}
    By renaming of the positional indices.
%    \qed
\end{proof}
%\begin{lemma}
%\label{lem:encshifttrace}
%    For any $\forall^2 \hyperltl$ formula $\forall \pi,\pi' \ldot \varphi$ over atomic propositions $\ap$, any finite traces $t,t' \in \Sigma^+$ with $|t|\geq 2, |t'|\geq2$ and $n,m \in \nat$ it holds:\\
%    \begin{equation*}
%        \encx{\ap}{t'[n\rangle} \models \constrn{\varphi}{t}{m} \Rightarrow \encx{\ap}{t'} \models \constrn{\varphi}{t}{m}
%    \end{equation*}
%\end{lemma}
In the following lemma and corollary, we show that the semantics of the next operators matches the finite LTL semantics.
\begin{lemma}
\label{lem:encshiftnext}
    For any $\forall^2$HyperLTL formula $\forall \pi,\pi' \ldot \varphi$ over atomic propositions $\ap$ and any finite traces $t,t' \in \Sigma^+$ it holds that
   % \begin{equation*}
   $
        %\begin{array}{ll}
             \enc{t'} \models \constr{\LTLnext \varphi}{t}
            \Leftrightarrow  \enc{t'} \models \constr{\varneg{\varphi}{1}}{t[1\rangle}
            \Leftrightarrow  \enc{t'[1\rangle} \models \constr{\varneg{\varphi}{0}}{t[1\rangle}.
        %\end{array}
        $
   % \end{equation*}
\end{lemma}
\begin{proof}
  Let $\varphi$, $t,t'$ be given.
  It holds that $\constr{\LTLnext \varphi}{t} =
  %\varneg{\varphi}{1} \land (\varneg{\varphi}{1} \imp \constr{\varneg{\varphi}{1}}{t[1\rangle}) = 
  	\constr{\varneg{\varphi}{1}}{t[1\rangle}$ by definition.
    As $\constr{\varneg{\varphi}{1}}{t[1\rangle}$ by construction does not contain any variables with positional index $0$, we only need to check satisfiablity with respect to $\enc{t'[1\rangle}$. Thus
    $
    %\begin{array}{ll}
         \enc{t'} \models \constr{ \LTLnext \varphi}{t}
        \Leftrightarrow  \enc{t'} \models \constr{\varneg{\psi}{1}}{t[1\rangle}
        \Leftrightarrow  \encn{t'[1\rangle}{1} \models \constr{\varneg{\varphi}{1}}{t[1\rangle}
        \xLeftrightarrow{\text{Lem}\ref{lem:encshift}}  \enc{t'[1\rangle} \models \constr{\varneg{\varphi}{0}}{t[1\rangle}
    %\end{array}
    $.
   % \qed
\end{proof}
\begin{corollary}
\label{lem:encshiftweaknext}
    For any $\forall^2$HyperLTL formula $\forall \pi,\pi' \ldot \varphi$ over atomic propositions $\ap$ and any finite traces $t,t' \in \Sigma^+$ it holds that
    % \begin{equation*}
    $
    %\begin{array}{ll}
    \enc{t'} \models \constr{\LTLweaknext \varphi}{t}
    \Leftrightarrow  \enc{t'} \models \constr{\varpos{\varphi}{1}}{t[1\rangle}
    \Leftrightarrow  \enc{t'[1\rangle} \models \constr{\varpos{\varphi}{0}}{t[1\rangle}.
    %\end{array}
    $
   % \end{equation*}
\end{corollary}
We will now state the correctness theorem, namely that our algorithm preserves the HyperLTL rewriting semantics.
%\subsection{Equisatisfiability $\ltl$ projection and Boolean constraints}
\begin{theorem}
\label{thm:equisat-ltlfin-sat}
    For every $\forall^2 $HyperLTL formula $\forall \pi,\pi' \ldot \varphi$ in negation normal form over atomic propositions $\ap$ and any finite trace $t \in \Sigma^+$ it holds that
    %\begin{equation*}
    $
        \forall t' \in \Sigma^+ \ldot t' \modelsltlfin \proj{\varphi}{\pi}{t} \Leftrightarrow \encx{\ap}{t'} \models \constr{\varphi}{t}.
    $
    %\end{equation*}
\end{theorem}
\begin{proof}
    By induction over the size of $t$.
    %\begin{itemize}
    %	\item
    Induction Base ($t = e$, where $e \in \Sigma$):
    %Define $e \defeq t[0]$ for ease of notation.\\
    We choose $t' \in \Sigma^+$ arbitrarily.
    %In order to proof the base of the induction we have to do a case distinction over the formula  $\varphi$ in the fashion of structural induction.
    %Structural Induction Base Cases:
    We distinguish by structural induction the following cases over the formula $\varphi$:
    %We begin with the base cases.
   % \vspace{-2ex}
    \begin{itemize}
        \item$a_{\pi}$:
            $\constr{a_{\pi}}{e} = (a_\pi)[\pi,e,0] = \top \text{ if, and only if, } a \in e$.
            Thus
            $\enc{t'} \models \constr{a_{\pi}}{e} \Leftrightarrow a \in e \Leftrightarrow t' \modelsltlfin \proj{a_{\pi}}{\pi}{e}$.
        %\item$\neg a_{\pi}$: analogously
        \item$a_{\pi'}$:
            $\constr{a_{\pi'}}{e} = (a_{\pi'})[\pi,e,0] = a_0$
            Thus
            $\enc{t'} \models \constr{a_{\pi'}}{e} \Leftrightarrow \enc{t'} \models a_0 \xLeftrightarrow{\text{def } enc} a \in t'[0] \Leftrightarrow t' \modelsltlfin a \xLeftrightarrow{\text{def } |^{\pi}} t' \modelsltlfin \proj{a_{\pi'}}{\pi}{e}$.
        \item $\neg a_{\pi}$ and $\neg a_{\pi'}$ are proven analogously.
    \end{itemize}
%\vspace{-2ex}
    The structural induction hypothesis states that
    $\forall t' \in \Sigma^+ \ldot  t' \modelsltlfin \proj{\psi}{\pi}{t} \Leftrightarrow \enc{t'} \models \constr{\psi}{t}$ (SIH1), where $\psi$ is a strict subformula of $\varphi$.
    %We proceed with the remaining cases.
    %\vspace{-2ex}
    \begin{itemize}
        \item$\varphi \lor \psi$:
            $
            %\begin{array}{ll}
                 t' \modelsltlfin \proj{(\varphi \lor \psi)}{\pi}{e}
                \Leftrightarrow (t' \modelsltlfin \proj{\varphi}{\pi}{e}) \lor (t' \modelsltlfin \proj{\psi}{\pi}{e})
                \xLeftrightarrow{\text{SIH1}} (\enc{t'} \models \constr{\varphi}{e}) \lor (\enc{t'} \models \constr{\psi}{e})
                \Leftrightarrow (\enc{t'} \models \varphi[\pi,e,0]) \lor (\enc{t'} \models \psi[\pi,e,0])
                \Leftrightarrow  \enc{t'} \models \varphi[\pi,e,0] \lor \psi[\pi,e,0]
                \xLeftrightarrow{\text{def } enc}  \enc{t'} \models (\varphi \lor \psi)[\pi,e,0]
                \Leftrightarrow  \enc{t'} \models \constr{\varphi \lor \psi}{e}.
            %\end{array}
            $
        %\item$\varphi \land \psi$: analogously
        \item$\LTLnext \varphi$:
            $\constr{\LTLnext \varphi}{e} = (\LTLnext\varphi)[\pi,e,0] = \varneg{\varphi}{0} \land (\varneg{\varphi}{0} \imp \constr{\varneg{\varphi}{0}}{\epsilon}) = \bot$.
            Thus
            $t' \modelsltlfin \proj{(\LTLnext \varphi)}{\pi}{e} = \bot \Leftrightarrow \enc{t'} \models \bot$.
        %\item$\LTLweaknext \varphi$: analogously
        \item$\varphi \LTLuntil \psi$:
            $\constr{\varphi \LTLuntil \psi}{e} = (\varphi \LTLuntil \psi)[\pi,e,0] = \psi[\pi,e,0] \lor (\varphi[\pi,e,0] \land \constr{\varneg{\varphi \LTLuntil \psi}{0}}{\epsilon}) =  \psi[\pi,e,0] = \constr{\psi}{e}$.
            Thus
            $t' \modelsltlfin \proj{(\varphi \LTLuntil \psi)}{\pi}{e} \modelsltlfin \proj{\psi}{\pi}{e} \xLeftrightarrow{\text{SIH1}} \enc{t'} \models \constr{\psi}{e}$.
        \item $\varphi \land \psi$, $\LTLweaknext \varphi$, and $\varphi \LTLrelease \psi$ are proven analogously.
    \end{itemize}
%\vspace{-2ex}
%\item
    Induction Step ($t = e \cdot t^*$, where $e \in \Sigma$ and $t^* \in \Sigma^+$):
    The induction hypothesis states that
    $\forall t' \in \Sigma^+ \ldot t' \modelsltlfin \proj{\varphi}{\pi}{t^*} \Leftrightarrow \enc{t'} \models \constr{\varphi}{t^*}$ (IH).
    %It is left to show that
    %$\forall t' \in \Sigma^+ \ldot t' \modelsltlfin \proj{\varphi}{\pi}{t} \Leftrightarrow \enc{t'} \models \constr{\varphi}{t}$.
    %Again we make use of structural induction over $\varphi$.
    We make use of structural induction over $\varphi$.
    %Note again, that all base cases are already covered as their proofs in the induction base were independent of the length of $t$.
    All base cases are covered as their proofs above are independent of $|t|$.
    %This leaves us with the proof obligation for $\varphi \lor \psi$, $\LTLnext \varphi$ and $\varphi \LTLuntil \psi$ (and their dual versions).
    %With the base cases done, we can make use of the structural induction hypothesis which states that
    The structural induction hypothesis states for all strict subformulas $\psi$ that
    $\forall t' \in \Sigma^+ \ldot t' \modelsltlfin \proj{\psi}{\pi}{t} \Leftrightarrow \enc{t'} \models \constr{\psi}{t}$.
    %, where $\psi$ is a strict subformula of $\varphi$.
    %\vspace{-2ex}
    \begin{itemize}
        \item$\varphi \lor \psi$:\\
            $
            \begin{array}{ll}
                & t' \modelsltlfin \proj{(\varphi \lor \psi)}{\pi}{t}
                \Leftrightarrow  t' \modelsltlfin \proj{\varphi}{\pi}{t} ~ \lor ~ t' \modelsltlfin \proj{\psi}{\pi}{t}\\
                \xLeftrightarrow{\text{SIH1}} & \enc{t'} \models \constr{\varphi}{t} ~ \lor ~ \enc{t'} \models \constr{\psi}{t}\\
                \xLeftrightarrow{t=e\cdot t^*} & \enc{t'} \models (\varphi[\pi,e,0] \land \bigwedge\limits_{\varboth{\varphi'}{1} \in \varphi[\pi,e,0]} \varboth{\varphi'}{1} \imp \constr{\varboth{\varphi'}{1}}{t^*})\\
                  & \lor ~ \enc{t'} \models (\psi[\pi,e,0] \land \bigwedge\limits_{\varboth{\psi'}{1} \in \varphi[\pi,e,0]} \varboth{\psi'}{1} \imp \constr{\varboth{\psi'}{1}}{t^*})\\
                %\Leftrightarrow & \ldots\\
                \xLeftrightarrow{\dag} & \enc{t'} \models (\varphi[\pi,e,0]  \lor \psi[\pi,e,0] )\\
                & \land \bigwedge\limits_{\varboth{\varphi'}{1} \in \varphi[\pi,e,0]} \varboth{\varphi'}{1} \imp \constr{\varboth{\varphi'}{1}}{t^*}
                \end{array}$
                
                $\begin{array}{ll}
                & \land \bigwedge\limits_{\varboth{\psi'}{1} \in \varphi[\pi,e,0]} \varboth{\psi'}{1} \imp \constr{\varboth{\psi'}{1}}{t^*}\\
                \Leftrightarrow & \enc{t'} \models (\varphi \lor \psi)[\pi,e,0]\\
                & \land \bigwedge\limits_{\varboth{\phi}{1} \in (\varphi \lor \psi)[\pi,e,0]} \varboth{\phi}{1} \imp \constr{\varboth{\phi}{1}}{t^*}\\
                \xLeftrightarrow{t=e\cdot t^*} & \enc{t'} \models \constr{\varphi \lor \psi}{t}\\
            \end{array}
            $\\
            $\dag$:
            %We rename the terms of our proof obligation:
            %$
            %\begin{array}{ll}
             %    \left[\enc{t'} \models \varphi[\pi,e,0] \land A\right] \lor \left[\enc{t'} \models \psi[\pi,e,0] \land B\right]
             %   \Longleftrightarrow  \enc{t'} \models \varphi[\pi,e,0] \lor \psi[\pi,e,0] \land A \land B
            %\end{array}\\
            %$
            %\begin{itemize}
                %\item[$\Leftarrow$]:
                %If we have a model $M$ the for $\enc{t'} \models \varphi[\pi,e,0] \lor \psi[\pi,e,0] \land A \land B$, we can infer that either $\varphi[\pi,e,0]$ or $\psi[\pi,e,0]$ are satisfied by $M$. So in fact $M$ satisfies $\enc{t'} \models \varphi[\pi,e,0] \land A$ or $\enc{t'} \models \psi[\pi,e,0] \land B$.
                $\Leftarrow$: trivial,
                $\Rightarrow$: Assume a model $M_\varphi$ for $\enc{t'} \models \varphi[\pi,e,0] \land A$. By construction, constraints by $\varphi$ do not share variable with constraints by $\psi$.
                %So we can extend model $M_\varphi$ in such a way that we get a model for $\enc{t'} \models \varphi[\pi,e,0] \lor \psi[\pi,e,0] \land A \land B$.
                We extend the model by assigning $\varboth{\psi'}{1}$ with $\bot$, for all $\varboth{\psi'}{1} \in \psi[\pi,e,0]$ and assigning the rest of the variables in $\psi[\pi,e,0]$ arbitrarily.
            %\end{itemize}
    	\item$\LTLnext \varphi$:
    	$
    	%\begin{array}{ll}
    	t' \modelsltlfin \proj{(\LTLnext \varphi)}{\pi}{t}
    	\Leftrightarrow  t' \modelsltlfin \LTLnext \proj{\varphi}{\pi}{t^*}
    	\Leftrightarrow  t'[1\rangle \modelsltlfin \proj{\varphi}{\pi}{t^*}
    	\xLeftrightarrow{\text{IH}}  \enc{t'[1\rangle} \models \constr{\varphi}{t^*}
    	%\xLeftrightarrow{\text{Lem}\ref{lem:encshift}} & \encn{t'[1\rangle}{1} \models \constrn{\varphi}{t}{1}\\
    	\xLeftrightarrow{\text{Lem}\ref{lem:encshiftnext}}  \enc{t'} \models \constr{\LTLnext \varphi}{t}{}.
    	%\end{array}
    	$
        %\item$\LTLweaknext \varphi$: analogously
        \item$\varphi \LTLuntil \psi$:\\
            $
            \begin{array}{ll}
                & t' \modelsltlfin \proj{(\varphi \LTLuntil \psi)}{\pi}{t}\\
                %\Leftrightarrow  t' \modelsltlfin \proj{\psi}{\pi}{t} \lor (\proj{\varphi}{\pi}{t} \land \varneg{\proj{(\varphi \LTLuntil \psi)}{\pi}{t^*}}{1})\\
                %\Leftrightarrow & t' \modelsltlfin \proj{\psi}{\pi}{t''} ~ \lor ~ \left[t' \modelsltlfin \proj{\varphi}{\pi}{t''} ~ \land ~ t' \modelsltlfin \LTLprev (\proj{(\varphi \LTLuntil \psi)}{\pi}{t}) \right]\\
                \Leftrightarrow & t' \modelsltlfin \proj{\psi}{\pi}{t} ~ \lor ~ \left[t' \modelsltlfin \proj{\varphi}{\pi}{t} ~ \land ~ t'[1\rangle \modelsltlfin \proj{(\varphi \LTLuntil \psi)}{\pi}{t^*} \right]\\
                %\xLeftrightarrow{\text{SIH1+IH}} & \enc{t'} \models \constr{\psi}{t} ~ \lor ~ \left[\enc{t'} \models \constr{\varphi}{t} ~ \land ~ \enc{t'[1\rangle} \models \constr{\varphi \LTLuntil \psi}{t^*}\right]\\
                \xLeftrightarrow{\text{SIH1+IH},\text{L}\ref{lem:encshiftnext}} & \enc{t'} \models \constr{\psi}{t} \\
                 &\lor \left[\enc{t'} \models \constr{\varphi}{t} ~ \land ~ \enc{t'} \models \constr{\varneg{\varphi \LTLuntil \psi}{1}}{t^*}\right]\\
                \Leftrightarrow & \enc{t'} \models (\psi[\pi,e,0] \land \bigwedge\limits_{\varboth{\psi'}{1} \in \psi[\pi,e,0]} \varboth{\psi'}{1} \imp \constr{\varboth{\psi'}{1}}{t^*})\\
                  & \lor ~ \left(
                \begin{array}{ll}
                    \enc{t'} & \models (\varphi[\pi,e,0] \land \bigwedge\limits_{\varboth{\varphi'}{1} \in \varphi[\pi,e,0]} \varboth{\varphi'}{1} \imp \constr{\varboth{\varphi'}{1}}{t^*})\\
                    \land ~ \enc{t'} & \models (\varneg{\varphi \LTLuntil \psi}{1} \land \varneg{\varphi \LTLuntil \psi}{1} \imp \constr{\varneg{\varphi \LTLuntil \psi}{1}}{t^*})
                \end{array}\right)\\
                %\Leftrightarrow & \ldots\\
                \xLeftrightarrow{\text{same as } \dag} & \enc{t'} \models ( \psi[\pi,e,0] \lor (\varphi[\pi,e,0] \land \varneg{\varphi \LTLuntil \psi}{1}))\\
                & \land \bigwedge\limits_{ \varboth{\psi'}{1} \in \psi[\pi,e,0]} \varboth{\psi'}{1} \imp \constr{\varboth{\psi}{1}}{t^*}\\
                & \land \bigwedge\limits_{\varboth{\varphi'}{1} \in \varphi[\pi,e,0]} \varboth{\varphi'}{1} \imp \constr{\varboth{\varphi'}{1}}{t^*}\\
                & \land ~ \varneg{\varphi \LTLuntil \psi}{1} \imp \constr{\varneg{\varphi \LTLuntil \psi}{1}}{t^*}\\
                \Leftrightarrow & \enc{t'} \models \varphi \LTLuntil \psi[\pi,e,0]\\
                & \land \bigwedge\limits_{{\varboth{\phi}{1}} \in \varphi \LTLuntil \psi[\pi,e,0]} \varboth{\phi}{1} \imp \constr{\varboth{\phi}{1}}{t^*}\\
                \Leftrightarrow & \enc{t'} \models \constr{\varphi \LTLuntil \psi}{t}\\
            \end{array}
            $
           \item $\varphi \land \psi$, $\LTLweaknext \varphi$, and $\varphi \LTLrelease \psi$ are proven analogously.
    \end{itemize}
%\vspace{-2ex}
%\end{itemize}
%\qed
\end{proof}

\begin{corollary}
\label{coro:same-lang-hyper-sat}
    For any $\forall^2$HyperLTL formula $\forall \pi,\pi' \ldot \varphi$ in negation normal form over atomic propositions $\ap$ and any finite traces $t,t' \in \Sigma^+$ it holds that
    %\begin{equation*}
       $t' \in \langx{\varphi}{t} \Leftrightarrow \encx{\ap}{t'} \models \constr{\hat{\varphi}}{t}$.
   % \end{equation*}
\end{corollary}
\begin{proof}
    %\begin{equation*}
    $
       % \begin{array}{ll}
             t' \in \langx{\varphi}{t}
            %\xLeftrightarrow{\text{Lem}\ref{lem:symm-closure}}  t' \in \langxx{\hat{\varphi}}{\pi}{t}
            \xLeftrightarrow{\text{Thm}\ref{thm:ltlfin-compat}}  t' \modelsltlfin \proj{\hat{\varphi}}{\pi}{t}
            \xLeftrightarrow{\text{Lem}\ref{thm:equisat-ltlfin-sat}}  \enc{t'} \models \constr{\hat{\varphi}}{t} .
       % \end{array}\\
    %\end{equation*}
    $
 %   \qed
\end{proof}

\begin{lemma}
\label{lem:earlyterm}
    For any $\forall^2$HyperLTL formula $\forall \pi,\pi' \ldot \varphi$ in negation normal form over atomic propositions $\ap$ and any finite traces $t,t' \in \Sigma^+$ it holds that
    %\begin{equation*}
        $\encx{\ap}{t'} \nmodels \constr{\varphi}{t} \Rightarrow \forall t'' \in \Sigma^+ \ldot t' \leq t'' \imp \encx{\ap}{t''} \nmodels \constr{\varphi}{t}$.
   %\end{equation*}
\end{lemma}
\begin{proof}
    We proof this via contradiction. We choose $t,t'$ as well as $\varphi$ arbitrarily, but in a way such that $\enc{t'} \nmodels \constr{\varphi}{t}$ holds.
    Assume that there exists a continuation of $t'$, that we call $t''$, for which $\enc{t''} \models \constr{\varphi}{t}$ holds.
    So there has to exist a model assigning truth values to the variables in $\constr{\varphi}{t}$, such that the constraint system is consistent.
    From this model we extract all assigned truths values for positional variables for position $|t'|$ to $|t''|-1$.
    As $t'$ is a prefix of $t''$, we can use these truth values to construct a valid model for $\enc{t'} \models \constr{\varphi}{t}$, which is a contradiction.
 %   \qed
\end{proof}

%\begin{lemma}
%    \label{lem:twoset}
%    For any $\forall^2 \hyperltl$ formula $\varphi$, $T \in \powerset(\Sigma^+)$ it holds:\\
%    \begin{equation*}
%        T \in \lang{\varphi} \Longleftrightarrow \forall t,t' \in T \ldot \{t,t'\} \in \lang{\varphi}\\
%    \end{equation*}
%\end{lemma}
%\newpage
\begin{corollary}
\label{coro:set-same-lang-hyper-sat}
    For any $\forall^2$HyperLTL formula $\forall \pi,\pi' \ldot \varphi$ in negation normal form over atomic propositions $\ap$ and any finite set of finite traces $T \in \powerset(\Sigma^+)$ and finite trace $t' \in \Sigma^+$ it holds that
    
   % \begin{equation*}
        \[t' \in \bigcap\limits_{t \in T}\langx{\varphi}{t} \quad \Longleftrightarrow \quad \encx{\ap}{t'} \models \bigwedge\limits_{t \in T}\constr{\hat{\varphi}}{t}.\]
    %\end{equation*}
\end{corollary}
\begin{proof}
    It holds that
    $\forall t,t' \in \Sigma^+ \ldot t \neq t' \imp \constr{\varphi}{t} \neq \constr{\varphi}{t'}$.
    %With same reasoning as above, it follows from \autoref{coro:same-lang-hyper-sat}.
    Follows with same reasoning as in earlier proofs combined with \autoref{coro:same-lang-hyper-sat}.
\end{proof}

\section{Experimental Evaluation}
%In the following section, we report on experimental results of our new $\forall^2$ HyperLTL monitoring tool. We briefly describe implementation details 
%\subsection{Implementation}
\tikzset{external/export=true}
\begin{figure}[t]
	%\begin{minipage}{0.8\textwidth}
	%\end{minipage}
	\resizebox{\textwidth}{!}{
		\input{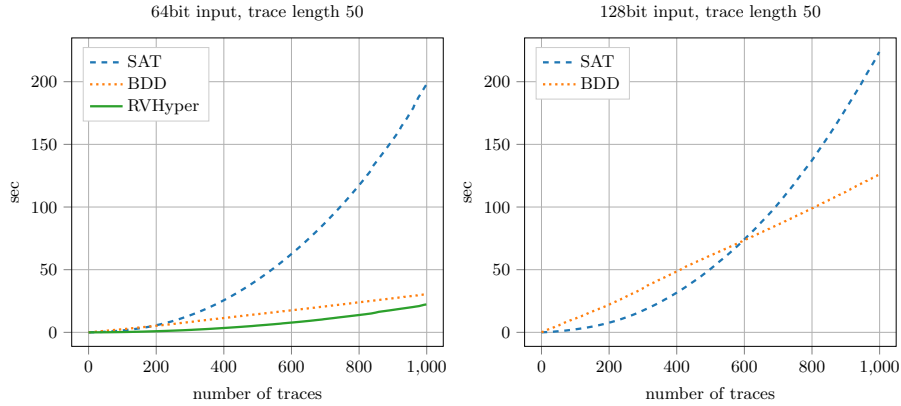}
	}
	\caption{Runtime comparison between RVHyper and our constraint-based monitor on a non-interference specification with traces of varying input size.}
	\label{fig:rand-ni-aps}
\end{figure}
\tikzset{external/export=false}
We implemented two versions of the algorithm presented in this paper.
The first implementation encodes the constraint system as a Boolean satisfiability problem (SAT), whereas the second one represents it as a (reduced ordered) binary decision diagram (BDD).
%In the following we will refer to these implementations as SAT and BDD, respectively.
%For the SAT based version we made use of the C interface of CryptoMiniSat~\cite{cms}, the BDD version is implemented using the BDD library CUDD~\cite{somenzi2009cudd}.
The formula rewriting is implemented in a Maude~\cite{conf/rta/ClavelDELMMT03} script.
The constraint system is solved by either CryptoMiniSat~\cite{conf/sat/SoosNC09} or CUDD~\cite{somenzi2009cudd}.
%Maude~\cite{maude} is a powerful rewriting engine, which has seen prior use in the context of LTL monitoring.
%We communicate with a Maude instance running on a separate process.
%This can essentially be abstracted as a function call to Maude, where we ask Maude to rewrite a formula plus an event for which Maude returns us a constraint.
%Under the hood this involves encoding and decoding of formulas and strings into strings and constraints, respectively.
%Both implementations sit on top of a common code base implemented in \textit{GO}.
%This includes the front-end providing the command line interface and parsing HyperLTL formulas, as well as the implementation of the algorithms, which also includes formula/constraint conversions and the communication layer to Maude.
%\subsection{Setup}
%For compilation we used GCC 8.1.1 (-O2) and go1.11.1.\\
%As third-party tool or libraries, we use Maude 2.7.1, CryptoMiniSat 5.5.0 and CUDD 3.0.0.
All benchmarks were executed on an Intel Core i5-6200U CPU @2.30GHz with 8GB of RAM. % running Linux Kernel 4.4.
%\subsection{Introduction}
%A major motivation for our constraint-based monitoring approach resides in the fact, that its runtime is essentially only linear in the number of traces as opposed to the quadratic runtime imposed by combinatorial solutions as applied in the tool RVHyper.
%In our evaluation we attempt to check if this performance hypothesis holds true for a variety of benchmarks.
%As we want to compare our new algorithms against existing solutions, i.e. RVHyper,
The set of benchmarks chosen for our evaluation is composed out of two benchmarks presented in earlier publications~\cite{conf/rv/FinkbeinerHST17,conf/tacas/FinkbeinerHST18} plus instances of \emph{guarded invariants} at which our implementations excels.
\smallparagraph{Non-interference.}
Non-interference~\cite{conf/sp/GoguenM82a,journals/jcs/McLean92} is an important information flow policy demanding that an observer of a system cannot infer any high security input of a system by observing only low security input and output.
Reformulated we could also say that all low security outputs $\vec{o}^{low}$ have to be equal on all system executions as long as the low security inputs $\vec{i}^{low}$ of those executions are the same:
%We can express this property in $\forall^2$ HyperLTL:\\
%$$\forall \pi, \pi' \ldot \bigwedge\limits_{o \in O_{low}} (o_{\pi} \leftrightarrow o_{\pi'}) \LTLweakuntil \bigvee\limits_{i \in I_{low}} (i_{\pi} \nleftrightarrow i_{\pi'})$$
$\forall \pi, \pi' \ldot ({\vec{o}^{low}_\pi} \leftrightarrow {\vec{o}^{low}_{\pi'}}) \LTLweakuntil ({\vec{i}^{low}_\pi} \nleftrightarrow {\vec{i}^{low}_{\pi'}}).$
This class of benchmarks was used to evaluated RVHyper~\cite{conf/tacas/FinkbeinerHST18}, an automata-based runtime verification tool for HyperLTL formulas.
%, in which RVHyper was first mentioned and evaluated.
We repeated the experiments and depict the results in Fig.~\ref{fig:rand-ni-aps}.
We choose a trace length of $50$ and monitored non-interference on $1000$ randomly generated traces, where we distinguish between a $64$ bit input (left) and an $128$ bit input (right).
For $64$ bit input, our BDD implementation performs comparably well to RVHyper, which statically constructs a monitor automaton.
For $128$ bit input, RVHyper was not able to construct the automaton in reasonable time. Our implementation, however, shows promising results for this benchmark class that puts the automata-based construction to its limit.
%, whereby we test following the setups.
%In order to see how the trace length influences the runtime, we evaluate with 20, 100, and 200 events per trace.
%Furthermore we evaluate the effects of the number of low input variables on the runtime of the tools, i.e., we choose a single input variable (as Finkbeiner et al. did in their evaluation) but also see how they do on 8 different input variable.
%In \autoref{fig:rand-ni} we compare the running times of our implementations with RVHyper on a set of Non-Interference benchmarks featuring the different parameters.

%From the plots we can deduce, that for small trace sizes (20) our implementations show comparable, in case of SAT, or even better results, in case of BDD, than RVHyper.
%By the nature of the Non-Interference property, the trace length 
%\todo{explain what we see in the plots}

%\tikzset{external/export=true}
%\begin{figure}[t]
    %\begin{minipage}{0.8\textwidth}
    %\end{minipage}
%    \resizebox{\textwidth}{!}{
 %       \input{./plots/ni.tex}
 %   }
  %  \caption{Non-Interference with trace lengths 20, 100, 200 and 8, 64 bit input size.}
   % \label{fig:rand-ni}
%\end{figure}
%\tikzset{external/export=false}

%\subsection{Non-Interference APs}
%\todo{explain results: BDD comparable with RVHyper, SAT almost not effected by increased atomic propositions}

\smallparagraph{Detecting Spurious Dependencies in Hardware Designs.}
\tikzset{external/export=true}
\begin{figure}[t]
	%\begin{minipage}{0.8\textwidth}
	%\end{minipage}
	\resizebox{\textwidth}{!}{
		\input{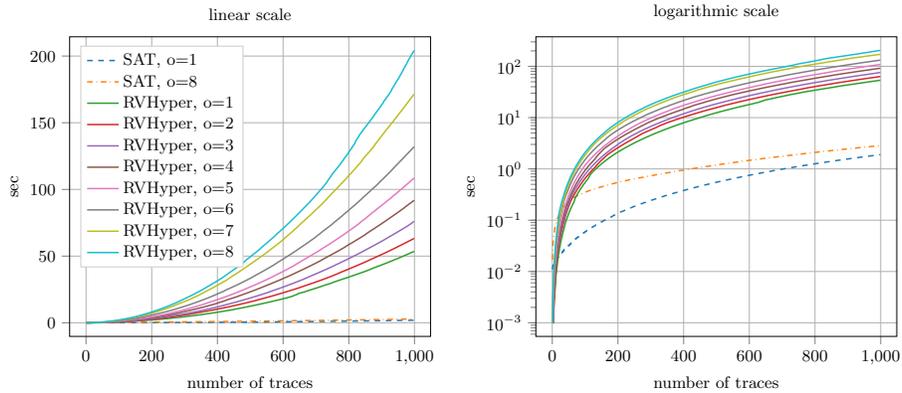}
	}
	\caption{Runtime comparison between RVHyper and our constraint-based monitor on the guarded invariant benchmark with trace lengths $20$, $20$ bit input size.}
	\label{fig:invariants}
\end{figure}
\tikzset{external/export=false}
%We considered the problem of detecting whether input signals influence output signals in hardware designs~\cite{conf/tacas/FinkbeinerHST18}.
The problem whether input signals influence output signals in hardware designs, was considered in~\cite{conf/tacas/FinkbeinerHST18}.
%We briefly describe the input specification.
%and the corresponding hardware designs, before reporting the results of those benchmarks.
%We write $\vec i \ninfluences \vec o$ to denote that the inputs $\vec i$ do not influence the outputs $\vec o$.
Formally, we specify this property as the following $\hyperltl$ formula:
$
\forall \pi_1 \forall \pi_2 \ldot
(\vec{o}_{\pi_1} \leftrightarrow \vec{o}_{\pi_2}) \LTLweakuntil (\overline{\vec{i}}_{\pi_1} \nleftrightarrow \overline{\vec{i}}_{\pi_2}),
$
where $\overline{\vec{i}}$ denotes all inputs except $\vec i$.
Intuitively, the formula asserts that for every two pairs of execution traces $(\pi_1,\pi_2)$ the value of $\vec{o}$ has to be the same until there is a difference between $\pi_1$ and $\pi_2$ in the input vector $\overline{\vec{i}}$, i.e., the inputs on which $\vec o$ may depend.
We consider the same hardware and specifications as in~\cite{conf/tacas/FinkbeinerHST18}. The results are depicted in Table~\ref{tbl:rvhyper-results}.
\begin{table}[t]
	\caption{Average results of our implementation compared to RVHyper on traces generated from circuit instances. Every instance was run $10$ times.}
		\label{tbl:rvhyper-results}
	\centering
    \setlength{\tabcolsep}{5pt}
    %{\scalebox{0.8}{
	\begin{tabular}{lrrrrrr}
		\hline \noalign{\smallskip}
		instance          & \#\,traces & length & time RVHyper & time SAT & time BDD  \\ \noalign{\smallskip}\hline\noalign{\smallskip}
		\textsc{xor1}     &       19   &    5   &     12ms     &     47ms &    49ms   \\
		\textsc{xor2}     &     1000   &    5   &  16913ms     &    996ms &  1666ms   \\
		counter1          &      961   &   20   &   9610ms     &   8274ms &   303ms   \\
		counter2          &     1353   &   20   &  19041ms     &  13772ms &   437ms   \\
		\textsc{mux1}     &     1000   &    5   &  14924ms     &    693ms &   647ms   \\
		\textsc{mux2}     &       80   &    5   &    121ms     &     79ms &    81ms   \\ \hline
	\end{tabular}
%}}
\end{table}
    %The results reassure what we have seen before, for many short traces the new constraint based solutions shows better performance than RVHyper.
    Again, the BDD implementation handles this set of benchmarks well.
    \tikzset{external/export=true}
    \begin{wrapfigure}{l}{0.52\textwidth}
    	%\begin{minipage}{0.8\textwidth}
    	%\end{minipage}
    	\centering
    	\resizebox{0.5\textwidth}{!}{
    		\input{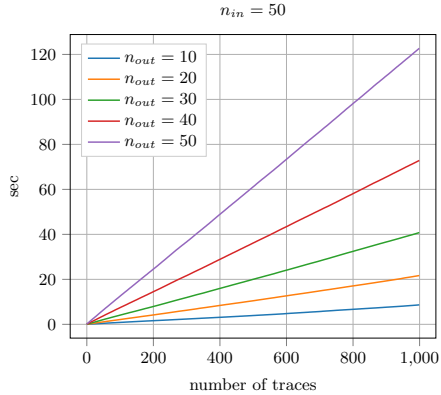}
    	}
    	\caption{Runtime of the SAT-based algorithm on the guarded invariant benchmark with a varying number of atomic propositions.}
    	\label{fig:many-aps}
    	\vspace{-17pt}
    \end{wrapfigure}
    \tikzset{external/export=false}
    The biggest difference can be seen between the runtimes for counter2.
    %when run with RVHyper and with BDD.
    %While RVHyper almost takes up 20 seconds on average, BDD does the job in 0.5 seconds.
    This is explained by the fact that this benchmark demands the highest number of observed traces, and therefore the impact of the quadratic runtime costs in the number of traces dominates the result.
    We can, in fact, clearly observe this correlation between the number of traces and the runtime on RVHyper's performance over all benchmarks.
    On the other hand our constraint-based implementations do not show this behavior.
    %BDD really only seem to struggle with \textsc{xor2}, which is explained by the formula structure and the seen traces not being perfectly suited for the binary decision diagram encoding.
    %SAT's runtime results seem to be dominated by the length of the trace, where the number of traces also effect the performance negatively, but rather in a linear fashion as can be seen with the counter use cases.

\smallparagraph{Guarded Invariants.}
We consider a new class of benchmarks, called \textit{guarded invariants}, which express a certain invariant relation between two traces, which are, additionally, guarded by a precondition.
% we want to evaluate on a different fragment of $\forall^2$ HyperLTL.
%The idea behind guarded invariants is that we want to express a certain invariant relation between two traces, which means, that on each position this relation has to be satisfied by the events of the two traces.
%Additionally, a precondition is allowed.
%\todo{insert interesting example}
%Let $P : \Sigma \rightarrow \bool$ be a predicate on events.
%A simple invariant is
%$\forall \pi, \pi' \ldot \LTLglobally (P(\pi) \leftrightarrow P(\pi'))$.
Fig.~\ref{fig:invariants} shows the results of monitoring an arbitrary invariant $P : \Sigma \rightarrow \bool$ of the following form:
%This guarded invariant can be expressed as:\\
$\forall \pi, \pi' \ldot \LTLeventually (\vee_{i \in I} i_{\pi} \nleftrightarrow i_{\pi'}) \rightarrow \LTLglobally (P(\pi) \leftrightarrow P(\pi'))$. Our approach significantly outperforms RVHyper on this benchmark class, as the conjunct splitting optimization, described in Section~\ref{alg}, synergizes well with SAT-solver implementations.

%\todo{present benchmark}
%\todo{explain benchmark results}

\smallparagraph{Atomic Proposition Scalability.}
While RVHyper is inherently limited in its scalability concerning formula size as the construction of the deterministic monitor automaton gets increasingly hard, the rewrite-based solution is not affected by this limitation.
To put it to the test we have ran the SAT-based implementation on guarded invariant formulas with up to 100 different atomic propositions.
Formulas have the form:
$\forall \pi, \pi' \ldot (\wedge_{i = 1}^{n_{in}} (in_{i,\pi} \leftrightarrow in_{i,\pi'})) \rightarrow \LTLglobally (\vee_{j = 1}^{n_{out}} (out_{j,\pi} \leftrightarrow out_{j,\pi'})),$
where $n_{in}, n_{out}$ represents the number of input and output atomic propositions, respectively.
Results can be seen in Fig.~\ref{fig:many-aps}.
Note that RVHyper already fails to build monitor automata for $|n_{in}+n_{out}| > 10$.

\section{Conclusion}
%We pursued the success story of rewrite-based monitoring algorithms for trace properties
We pursued the success story of rewrite-based monitors for trace properties by applying the technique to the runtime verification problem of Hyperproperties. We presented an algorithm that, given a $\forall^2$HyperLTL formula, incrementally constructs constraints that represent requirements on future traces, instead of storing traces during runtime.
Our evaluation shows that our approach scales in parameters where existing automata-based approaches reach their limits.
\smallparagraph{Acknowledgments.}
We thank Bernd Finkbeiner for his valuable feedback on earlier versions of this paper.%\todo{phrasing}

\bibliographystyle{splncs04}
\bibliography{main}

\end{document}